\documentclass[runningheads]{llncs}
\usepackage[T1]{fontenc}
\usepackage{graphicx,xcolor}

\usepackage{hyperref}
\hypersetup{
	colorlinks=true,
	linkcolor=blue
}

\usepackage{amssymb}
\usepackage{amsmath}
\usepackage{mathrsfs}
\newtheorem{assumption}{Assumption}
\usepackage{comment}
\usepackage{array}
\newcolumntype{C}[1]{>{\centering\arraybackslash}p{#1}}
\usepackage[boxruled,ruled,vlined]{algorithm2e}
\SetAlgorithmName{Protocol}{List of Protocols}{protocol}
\usepackage{bbm}
\usepackage{enumitem}
\usepackage{tcolorbox}
\usepackage{boxedminipage} 
\usepackage{caption}
\usepackage{xurl}
\usepackage{marvosym}
\begin{document}
\title{FlexProofs: A Vector Commitment with Flexible Linear Time for Computing All Proofs}

\author{Jing Liu \and
	Liang Feng Zhang\textsuperscript{(\Letter)}
}
\authorrunning{J. Liu and L. F. Zhang}

\institute{ShanghaiTech University, Shanghai, China \\ \email{\{liujing1,zhanglf\}@shanghaitech.edu.cn} }

%\author{First Author\inst{1}\orcidID{0000-1111-2222-3333} \and
%Second Author\inst{2,3}\orcidID{1111-2222-3333-4444} %\and
%Third Author\inst{3}\orcidID{2222--3333-4444-5555}}
%\institute{Princeton University, Princeton NJ 08544, USA \and
%Springer Heidelberg, Tiergartenstr. 17, 69121 Heidelberg, Germany
\maketitle              % typeset the header of the contribution
\begin{abstract}
%The abstract should briefly summarize the contents of the paper in 150--250 words.

In this paper, we introduce FlexProofs, a new {\em vector commitment (VC)} scheme that achieves two key properties: (1) the prover can generate all individual opening proofs for a vector of size $N$ in optimal time ${\cal O}(N)$, and there is a flexible batch size parameter $b$ that can be increased to further reduce the time to generate all proofs; and (2) the scheme is directly compatible with a family of zkSNARKs that encode their input as a multi-linear polynomial.
As a critical building block, we propose the first {\em functional commitment (FC)} scheme for multi-exponentiations with batch opening. 
Compared with HydraProofs, the only existing VC scheme that computes all proofs in optimal time ${\cal O}(N)$ and is directly compatible with zkSNARKs, FlexProofs may speed up the process of
generating  all proofs, if the parameter $b$ is properly chosen. Our experiments show that for $N=2^{16}$ and $b=\log^2 N$, FlexProofs can be  $6\times$ faster than HydraProofs.
Moreover, when combined with suitable zkSNARKs, FlexProofs enable practical applications such as verifiable secret sharing and verifiable robust aggregation.

%A key feature of our VC scheme is the tunable batch size $b$: increasing $b$ reduces the time to generate all proofs at the cost of larger proofs and longer verification time. 

\keywords{Vector Commitment  \and Functional Commitment \and zkSNARKs.}
\end{abstract}
\section{Introduction}
Zero-knowledge succinct non-interactive arguments of knowledge (zkSNARKs) \cite{K92,G16} allow a prover to convince a verifier with input ${\bf m}$  that ${\bf y}$ is the output of a computation ${\cal C}({\bf m}, {\bf w})$, where $\cal C$ is typically represented as  an arithmetic circuit and $\bf w$ may be a private input held by the prover.
While zkSNARKs have been widely used to secure real-world applications \cite{LXZ+24,SCG+14,LXZ21,APK+24,GAZ+22}
and  are well known for enabling efficient verification  by end-users, it is less clear how to apply them in a {\em multi-user} setting  where $N$ participants $\{{\cal V}_i\}_{i=0}^{N-1}$, each holding a private input $m_i$, 
delegate the computation  of ${\cal C}((m_0,\ldots,m_{N-1}),{\bf w})$ to  a service provider.
Such multi-user settings underlie numerous applications such as distributed machine learning \cite{MMR+17}, crowdsourcing \cite{H+06}, secret sharing \cite{S79}, and collaborative filtering \cite{SFH+07}, where it may be crucial to address the risk of a dishonest service provider  \cite{LYW+23,F87}.
The traditional zkSNARKs are not directly suitable for this setting because they assume each
verifier knows the entire input ${\bf m}=(m_0,\ldots,m_{N-1})$, which however is not true in the multi-user setting.  

Recently, Pappas,  Papadopoulos, and Papamanthou \cite{PPP25} proposed a method of adapting zkSNARKs to the multi-user setting  by combining them with {\em vector commitment (VC)} \cite{CF13}.
In their solution, the prover  uses a VC scheme to commit to the input vector ${\bf m}=(m_0,\ldots,m_{N-1})$ and runs a zkSNARK for ${\cal C}({\bf m},{\bf w})$ (similar to commit-and-prove zkSNARKs \cite{FFG+16,CFQ19}). 
Since VCs support openings of the vector at chosen indices, the prover can additionally provide each verifier ${\cal V}_i$ with a proof that its input $m_i$ is indeed the $i$-th element of the committed vector used in the zkSNARK, ensuring both computation correctness and individual data inclusion. 
This scenario requires a VC scheme that  can {\em efficiently generate all proofs} and  is 
{\em directly compatible with zkSNARKs}.

Among the existing VC schemes, the Merkle trees \cite{M87} support generating all proofs in ${\cal O}(N)$ time, but are inefficient when integrated into zkSNARKs.
Some modern VC schemes, e.g., those based on elliptic curves  \cite{SCP+22,WUP23,TAB+20,GRW+20,LZ22} or error-correcting codes \cite{ZXH+22,ZLG+24}, are directly compatible with many zkSNARKs  \cite{XZZ+19,S20,GWC19}, yet require ${\cal O}(N^2)$ time to naively produce all proofs.
Some constructions \cite{SCP+22,TAB+20,LZ22,ZXH+22,ZLG+24} may reduce the time cost to ${\cal O}(N\log N)$.
To the best of our knowledge, HydraProofs \cite{PPP25} is the only existing VC scheme that can generate all opening proofs in ${\cal O}(N)$ time and is directly compatible with zkSNARKs based on multi-linear polynomials \cite{XZZ+19,S20,CBB+23}.
In particular, the ${\cal O}(N)$ time cost is incurred by ${\cal O}(N)$ field operations and ${\cal O}(N)$   cryptographic operations (e.g., group exponentiations or pairings). 
In this paper, we propose FlexProofs, a new VC scheme that preserves zkSNARK compatibility and generates all proofs using only ${\cal O}(N)$ field operations and ${\cal O}(N/b + \sqrt{N}\log N)$ heavy cryptographic operations, where $b$ is the batch size (ranging from 1 to $\sqrt{N}$). Compared with HydraProofs, FlexProofs requires fewer heavy cryptographic operations for larger $b$, enabling more efficient generation of all proofs.

\subsection{Our Contributions}
\vspace{2mm}
\noindent
{\bf Functional Commitment Scheme for Multi-Exponentiations with Batch Opening.}
To construct FlexProofs, we propose the first {\em functional commitment (FC)} scheme for multi-exponentiations with batch opening. 
In the scheme, the prover commits to a vector and later generates a batch proof for multiple computations over the vector; any verifier can check computation results against the commitment. 
The scheme achieves constant-size commitments and logarithmic-size proofs, and is proven secure in the algebraic group model (AGM) and the random oracle model (ROM) under the $n$-ASDBP and $q$-SDH assumptions; our implementation confirms the high efficiency of batch opening.

\vspace{2mm}
\noindent
{\bf FlexProofs.}
We introduce FlexProofs, a new {\em vector commitment (VC)} scheme that achieves two key properties. First, the prover can generate all proofs for a vector of size $N$ in optimal time ${\cal O}(N)$, which is incurred by ${\cal O}(N)$ field operations and ${\cal O}(N/b + \sqrt{N}\log N)$ cryptographic operations, where batch size $b$ ranges from $1$ to $\sqrt{N}$.
Second, the scheme is directly compatible with a family of zkSNARKs that encode their input as a multi-linear polynomial \cite{XZZ+19,S20,CBB+23}.
FlexProofs builds on the proposed FC scheme and an existing {\color{black}polynomial commitment (PC)} scheme, and is proven correct and secure {\color{black}(see below for the detailed design rationale)}.
Compared with HydraProofs, the only existing VC scheme that computes all proofs in time ${\cal O}(N)$ (incurred by ${\cal O}(N)$ field operations and ${\cal O}(N)$ cryptographic operations) and is directly compatible with zkSNARKs, FlexProofs requires fewer cryptographic operations for larger $b$, enabling more efficient proof generation.
Our experiments show that for $N=2^{16}$ and $b=\log^2 N$, FlexProofs generates all proofs about $6\times$ faster than HydraProofs, while achieving similar verification time and slightly larger proofs.
Finally, combining our VC scheme with suitable zkSNARKs enables practical applications such as verifiable secret sharing and verifiable robust aggregation.

%A key feature of our design is the tunable batch size $b$: increasing $b$ reduces the number of cryptographic operations (such as group exponentiations or pairings) in {\sf OpenAll}, thereby improving its efficiency, at the cost of larger proof size and longer verification time.

{\color{black}
\vspace{2mm}
\noindent
{\bf Design Rationale for FlexProofs.}
FlexProofs adopts a clear two-layer structure: each sub-vector is first committed using an existing
PC scheme, and the resulting PC commitments are then committed with our proposed FC scheme. 
Correspondingly, the proof generation process also naturally splits into an FC layer and a PC 
layer, which allow us to do optimization independently: the FC layer employs our batch-opening 
FC scheme to lower the cost, while the PC layer applies ideas similar to HydraProofs to further 
reduce the cost. Overall, the two-layer structure combining FC and PC, together with our batch-opening FC scheme, enables FlexProofs to outperform purely PC-based HydraProofs.}
\subsection{Applications}
\vspace{2mm}
\noindent
{\bf Verifiable Secret Sharing.}
A real-world application involving multiple users' data is {\em verifiable secret sharing (VSS)}. 
At a high level, a secret sharing scheme \cite{S79} allows a dealer to split a secret value $s$ into $N$ shares such that any $t{+}1$ of them can reconstruct $s$, while any subset of at most $t$ reveals nothing about it. 
Verifiability protects the receivers of shares against a dishonest dealer who may issue malformed or inconsistent shares \cite{F87,CGM+85}. 
The combination of VCs and zkSNARKs yields a VSS scheme where each receiver can verify that all shares form a valid sharing of $s$ and that its own share is consistent with this sharing.

\vspace{2mm}
\noindent
{\bf Verifiable Robust Aggregation.}
{\color{black}
Another application is {\em federated learning (FL)} \cite{BBL11,MMR+17,KMY+16}, where multiple clients train local models and send gradients to an {\em aggregator}, who combines them into a global model. This process is repeated iteratively until the model converges.
Due to its decentralized nature, FL is vulnerable to misbehaving clients that submit poisoned or low-quality gradients, motivating a long line of work on {\em robust aggregation} to mitigate the impact of such adversarial inputs \cite{BNL12,KT22,CFL+20,MMM+22,FYB18}.
However, the security of robust aggregation relies on an honest aggregator, and these guarantees collapse when the aggregator is untrusted.
}
%Due to its decentralized nature, FL is vulnerable to misbehaving clients that may submit poisoned or low-quality gradients to degrade the model \cite{BNL12,KT22}. %48,49,50
%To address this, a long line of research has focused on {\em robust aggregation}, which aims to mitigate the impact of such adversarial inputs \cite{CFL+20,MMM+22,FYB18}. %51,52,53,54,55
%While robust aggregation mitigates the impact of malicious client updates, its guarantees rely on the aggregator honestly executing the procedure. If the aggregator itself is untrusted, these guarantees collapse. 
In practice, aggregators often have strong incentives to misuse their power for personal benefit.
For instance, in federated recommendation systems \cite{YTZ+20,TLZ+20}, %56,57,58
an aggregator can tamper with the model to promote its products, influence markets, or push its political agendas \cite{CNBC18,Guardian21,HGP+19}; %60,61,62,63
in FL-as-a-Service \cite{KKP20}, it may delay convergence for monetary gain \cite{XLL+19}. %65,66,67
To safeguard FL against both a misbehaving aggregator and malicious clients, Pappas et al.~\cite{PPP25} introduce the notion of {\em verifiable robust aggregation (VRA)}.
And by applying the combination of VCs and zkSNARKs to existing robust aggregation algorithms such as FLTrust \cite{CFL+20}, one can obtain such a VRA scheme.

\subsection{Related Work}
\vspace{2mm}
\noindent
{\bf Functional Commitments.}
Existing FC schemes target polynomials \cite{CFT22,WW23}, linear functions \cite{LRY16,LM19,LFG21,CNR+22,CFK+22}, monotone span programs \cite{CFT22}, or semi-sparse polynomials \cite{LP20}.  
{\color{black}
Prior FC schemes for Boolean or arithmetic circuits cannot handle multi-exponentiation
$\prod_{i=0}^{n-1}{\bf A}[i]^{{\bf b}[i]}({\bf A}\in \mathbb{G}_1^n, {\bf b}\in \mathbb{F}_p^n)$, where $\mathbb{G}_1$ is the source group of a bilinear group. Some schemes fail because they cannot commit to vectors over $\mathbb{G}_1$: for instance, \cite{PPS21,CP22} commit to vectors over $\mathbb{Z}_q$; \cite{BCF+22} commits to vectors over a commutative ring; \cite{WW24} commits to vectors over finite rings; and \cite{BNO21} commits to arithmetic circuits in a field. Other schemes fail for a different reason: FC schemes \cite{WW23,WW23+,SLH24} targeting Boolean functions $f:\{0,1\}^l\rightarrow \{0,1\}$ operate over bitstrings and fixed Boolean predicates, which cannot represent or efficiently compute exponentiations over $\mathbb{G}_1$.	
}

%For the multi-exponentiation $\prod_{i=0}^{n-1}{\bf A}[i]^{{\bf b}[i]}$, representing it as a Boolean or arithmetic circuit yields linear size (and depth ${\cal O}(\log n)$ for Boolean circuits), but existing FC schemes only support circuits of bounded size, depth, or width \cite{PPS21,CP22,WW23,WW23+,SLH24,BCF+22,BNO21} and thus cannot handle large $n$.  
%Moreover, the FC scheme \cite{WW24} operates over $\mathbb{Z}_p$ and cannot handle multi-exponentiation over {\color{blue} the source group $\mathbb{G}_1$ of a bilinear group}.

\vspace{2mm}
\noindent
{\bf Vector Commitments.}
Merkle trees \cite{M87} support generating all proofs in ${\cal O}(N)$ time.
However, they are not well suited for zkSNARKs, since incorporating this VC pre-image into existing zkSNARKs would essentially require reconstructing the entire tree within the zkSNARK arithmetic circuit, which is highly inefficient even with SNARK-friendly hash functions \cite{GKR+21}.
Some modern VC schemes, such as those based on elliptic curves \cite{SCP+22,WUP23,TAB+20,GRW+20,LZ22} or error-correcting codes \cite{ZXH+22,ZLG+24}, are directly compatible with many zkSNARKs \cite{XZZ+19,S20,GWC19} since they use the same data encoding. However, unlike Merkle trees, generating all $N$ proofs is inefficient: the naive cost is ${\cal O}(N^2)$, and even improved constructions \cite{SCP+22,TAB+20,LZ22,ZXH+22,ZLG+24} require ${\cal O}(N \log N)$ time.
To the best of our knowledge, HydraProofs \cite{PPP25} is the only existing VC scheme that can generate all proofs in ${\cal O}(N)$ time while being directly compatible with zkSNARKs based on multi-linear polynomials \cite{XZZ+19,S20,CBB+23}.
We do not consider lattice-based VC techniques \cite{PST+13,PPS21,WW23}, as the corresponding lattice-based zkSNARKs remain largely impractical.

%For completeness, we note other methods for loading (parts of) a dataset into a zkSNARK, such as accumulator-based approaches \cite{OWW+20,XHT+24}, which incur super-linear costs for partial proofs or require opening the pre-image within the circuit. We also exclude lattice-based VC techniques \cite{PST+13,PPS21,WW23}, as the corresponding lattice zkSNARKs are still largely impractical.

\section{Preliminaries}
{\color{black}
\vspace{2mm}
\noindent
{\bf Notation.}}
For any integer $n>0$, let
$[0, n)=\{0, 1, \ldots, n-1\}$. 
Let ${\bf m}$ be a vector of length $n$. For any $i\in [0, n)$, we denote by ${\bf m}[i]$ the $i$-th element of ${\bf m}$. For any $I\subseteq [0, n)$, we denote ${\bf m}[I]=({\bf m}[i])_{i\in I}$. Besides, we denote 
$$
{\bf m}_L=({\bf m}[0], \ldots, {\bf m}[n/2-1]), \quad
{\bf m}_R =({\bf m}[n/2], \ldots, {\bf m}[n-1]), 
$$
We denote by $\mathbb{G}$ a {\color{black} multiplicative} group, and by $\mathbb{F}_p$ the finite field of prime order $p$.
%We use bolded lower-case letters such as ${\bf a}$ to denote vectors over a finite field $\mathbb{F}_p$, and use bolded upper-case letters such as ${\bf A}$ to denote vectors over a group $\mathbb{G}$.
For a vector ${\bf A}\in \mathbb{G}^n$, a vector ${\bf b}\in \mathbb{F}_p^n$ and a scalar $x \in \mathbb{F}_p$, we denote
$$
{\bf A}^x=({\bf A}[0]^x, \ldots, {\bf A}[n-1]^x), 
x{\bf b}=(x{\bf b}[0],\ldots,x{\bf b}[n-1]),
\langle {\bf A}, {\bf b}\rangle =\prod_{i\in [0, n)}{\bf A}[i]^{{\bf b}[i]}.
$$
{\color{black}Here, $\langle {\bf A}, {\bf b}\rangle$ compactly denotes the multi-exponentiation of ${\bf A}$ by ${\bf b}$, which our FC scheme is designed to support.}
For two vectors ${\bf A}, {\bf A}^\prime \in \mathbb{G}^n$, we denote ${\bf A}\circ {\bf A}^\prime = ({\bf A}[0]{\bf A}^\prime[0], \ldots, {\bf A}[n-1]{\bf A}^\prime[n-1])$.
Let ${\bf u}_i$ be the $i$-th unit vector, with 1 at the $i$-th position and zeros elsewhere.
We define ${\sf Bin}(i) = (i_{\ell-1}, \ldots, i_0)$ as the {\it bit decomposition} of an $\ell$-bit integer $i$ if $i = \sum_{k \in [0, \ell)}i_{k}2^{k}.$
Any vector $\mathbf{m} \in \mathbb{F}_p^n$ can be encoded as a multilinear polynomial 
$f_{\mathbf{m}}: \mathbb{F}_p^{\log n} \to \mathbb{F}_p$ as below:
\begin{equation}\label{eq:multi-extension}
f_{\mathbf{m}}(\mathbf{x}) = 
\sum_{i \in [0, n)} \mathbf{m}[i] \prod_{k \in [0, \log n)} \left( i_k x_k  + (1 - i_k)(1 - x_k) \right),
\end{equation}
where $\mathbf{x} = (x_{\log n - 1}, \ldots, x_0)$, ${\sf Bin}(i) = (i_{\log n - 1}, \ldots, i_0)$.
We refer to $f_{\mathbf{m}}$ as the {\it multi-linear extension} of ${\bf m}$.
For any finite set $S$, we denote by $s\leftarrow S$ the process of choosing $s$ uniformly  from $S$. We denote by
$y\leftarrow {\sf Alg}(x)$ the process of running an algorithm ${\sf Alg}$ on an input $x$ and assigning the output to $y$.
We say that a function $\epsilon(\lambda)$ is {\it negligible} in $\lambda$ and denote $\epsilon(\lambda)={\sf negl}(\lambda)$, if $\epsilon(\lambda)=o(\lambda^{-c})$ for all $c>0$.
\begin{comment}
\begin{align*} 
{\bf A}^x&=(A_0^x, \ldots, A_{n-1}^x);~~~~
\hspace{3mm}\langle {\bf A}, {\bf b}\rangle =\prod_{i\in [0, n)}A_i^{b_i};\hspace{3mm}{\bf A}[V]=(A_i)_{i\in V}\\
{\bf A}_L&=(A_0, \ldots, A_{n/2-1}); 
\hspace{3mm}{\bf A}_R =(A_{n/2}, \ldots, A_{n-1}).
\end{align*} 
\end{comment}
Our security proofs are in the
{\em random oracle model} (ROM), formalized in \cite{BR93}: we model a cryptographic hash function as a truly random function, accessible to all parties only via oracle queries. Specifically, we use two random oracles $H, H^\prime: \{0, 1\}^*\rightarrow \mathbb{F}_p$.

\begin{comment}
We denote by ${\sf BG}(1^\lambda)$
a probabilistic polynomial-time (PPT) algorithm that takes a security parameter $\lambda$ as input and
outputs a {\it bilinear context} ${\bf bg}=(p, \mathbb{G}_1, \mathbb{G}_2, \mathbb{G}_{\rm T}, e,   g_1,   g_2)$, where  $\mathbb{G}_1=\langle g_1\rangle, \mathbb{G}_2=\langle g_2\rangle$ and $\mathbb{G}_{\rm T}$ are cyclic groups of prime order $p\approx 2^\lambda$, and $e:\mathbb{G}_1\times \mathbb{G}_2\rightarrow \mathbb{G}_{\rm T}$ is a    {\it pairing} with the following properties:
(1) {\it efficiently computable}: There is a polynomial-time algorithm to compute
$e(u,v)$ for all $u\in \mathbb{G}_1$ and $v\in \mathbb{G}_2$;
(2) {\it non-degenerate}: $e(g_1,g_2)$ is  a  generator of   $\mathbb{G}_{\rm T}$.
(3) {\it bilinear}: $ e(u^\alpha,v^\beta)=e(u,v)^{\alpha\beta}$ for any $u\in \mathbb{G}_1, v\in \mathbb{G}_2$ and $\alpha,\beta\in \mathbb{Z}_p$.
We  assume {\it Type-3} bilinear groups $(\mathbb{G}_1,\mathbb{G}_2)$ where
{\it no} efficiently computable homomorphisms exist between   $\mathbb{G}_1$ and $\mathbb{G}_2$.
For any ${\bf A}\in \mathbb{G}_1^n, {\bf B}\in \mathbb{G}_2^n$, we denote 
${\bf A}*{\bf B}=\prod_{i\in [0, n)}e(A_i, B_i).$
\end{comment}

%\subsection{Bilinear group}

{\color{black}
\vspace{2mm}
\noindent
{\bf Bilinear Group.}}
We denote by  ${\sf BG}(1^\lambda)$ a bilinear group generator that takes 
a security parameter $\lambda$ as input and outputs a bilinear group context
${\bf bg}=(p, \mathbb{G}_1, \mathbb{G}_2, \mathbb{G}_{\rm T}, e,   g_1, g_2)$, where $\mathbb{G}_1
=\langle g_1\rangle, \mathbb{G}_2=\langle g_2 \rangle$ and $\mathbb{G}_{\rm T}$ are groups of prime order $p$, and  $e:\mathbb{G}_1\times \mathbb{G}_2\rightarrow \mathbb{G}_{\rm T}$ is a {\em pairing} such that 
$ e(u^\alpha,v^\beta)=e(u,v)^{\alpha\beta}$ for all
$u\in \mathbb{G}_1, v\in \mathbb{G}_2$ and $\alpha,\beta\in \mathbb{Z}_p$.
We  assume {\it Type-3} bilinear groups where
{\it no} efficiently computable homomorphisms exist between   $\mathbb{G}_1$ and $\mathbb{G}_2$.
For any ${\bf A}\in \mathbb{G}_1^n, {\bf B}\in \mathbb{G}_2^n$, we denote 
${\bf A}*{\bf B}=\prod_{i\in [0, n)}e({\bf A}[i], {\bf B}[i]).$

\begin{comment}
\subsection{Merkle Tree}
Merkle tree \cite{M87} has been widely used for the vector commitment because of its simplicity and efficiency. The prover time is linear in the size of the vector while the verifier time and proof size are logarithmic in the size of the vector. 
Given a vector of ${\bf r}=(r_0, \ldots, r_{N-1})$, it consists of three algorithms:	
${\sf rt}\leftarrow {\sf MT.Commit}({\bf r})$, ${\sf path}_i \leftarrow{\sf MT.Open}(i,{\bf r})$ and $\{1,0\} \leftarrow {\sf MT.Verify}({\sf rt},i,r_i,{\sf path}_i)$.
\end{comment}

\subsection{Functional Commitments with Batch Openings}
Libert et al. \cite{LRY16} introduced a {\it functional commitment} (FC) model for linear functions over finite fields. 
Multi-exponentiations over bilinear groups naturally extends these linear functions: for ${\bf A} \in \mathbb{G}_1^n$ and ${\bf b} \in \mathbb{F}_p^n$, $\langle {\bf A}, {\bf b}\rangle=\prod_{i=0}^{n-1}{\bf A}[i]^{{\bf b}[i]}$.
We extend the FC model of \cite{LRY16} to support multi-exponentiations and batched evaluations on multiple field vectors. 
An FC scheme ${\sf FC}=({\sf Setup}, {\sf Commit}, {\sf BOpen}, {\sf BVerify})$ in our model consists of four algorithms:
\begin{itemize}
	\it
	\item
	${\sf FC.Setup}(1^\lambda, 1^n)\rightarrow {\sf pp}$:
	Given the security parameter $\lambda$ and the vector size $n$, outputs public parameters ${\sf pp}$,  an implicit input to all remaining algorithms.	
	\item
	${\sf FC.Commit}({\bf A})\rightarrow C$:
	Given a vector ${\bf A}\in\mathbb{G}_1^n$, outputs a commitment $C$.
	
	%\item
	%${\sf FC.Open}({\bf A}, {\bf b})\rightarrow \pi_y$:
	%Given ${\bf A}$ and a vector ${\bf b}\in \mathbb{F}_p^n$, the algorithm outputs a proof $\pi_y$ for $y=\langle {\bf A}, {\bf b}\rangle$.
	
	\item ${\sf FC.BOpen}(C, {\bf A}, \{{\bf b}^{(i)}\}_{i\in [0, t)}, \{y_{i}\}_{i\in [0, t)}) \to \pi_y$:
	Given a commitment $C$ (to $\bf A$), a vector ${\bf A}$, and vectors ${\bf b}^{(0)},\ldots,{\bf b}^{(t-1)}\in \mathbb{F}_p^n$, outputs a batch proof $\pi_y$ for $\{y_i=\langle {\bf A}, {\bf b}^{(i)}\rangle\}_{i\in [0, t)}$.
	
	%\item
	%${\sf FC.Verify}(C, {\bf b}, y, \pi_y)\rightarrow\{0, 1\}$:
	%The algorithm verifies the proof $\pi_y$ that $y=\langle {\bf A}, {\bf b}\rangle$ against the commitment $C$, and outputs 1 (accept) or 0 (reject). 
	
	\item ${\sf FC.BVerify}(C, \{{\bf b}^{(i)}\}_{i\in [0, t)}, \{y_{i}\}_{i\in [0, t)}, \pi_y) \to \{0,1\}$: Verifies the batch proof $\pi_y$ for $\{y_i = \langle {\bf A}, {\bf b}^{(i)}\rangle\}_{i\in[0,t)}$ against $C$ and outputs 1 (accept) or 0 (reject).
\end{itemize}

%While Libert et al. \cite{LRY16} requires that an FC scheme should be correct, function binding and hiding, in this paper we will ignore the hiding property, which is not needed in the our proposed applications.
An FC scheme is {\em correct} if ${\sf FC.BVerify}$ always outputs 1, provided that all algorithms are correctly executed. 
\begin{definition}
	[\bf Correctness]\label{df:FCCor}
	For any security parameter $\lambda$, any integer $n>0$,
	any vector ${\bf A}\in \mathbb{G}_1^n$, any integer $t>0$, any vectors ${\bf b}^{(0)},\ldots,{\bf b}^{(t-1)} \in \mathbb{F}_p^n$, define $y_i = \langle \mathbf{A}, \mathbf{b}^{(i)} \rangle$ for all $i\in [0, t)$, then 
	$$
	\Pr
	\begin{bmatrix}
	{\sf pp} \leftarrow {\sf FC.Setup}(1^\lambda, 1^n),
	C \leftarrow {\sf FC.Commit}({\bf A}),\\
	\pi_{y} \leftarrow  {\sf FC.BOpen}(C, {\bf A}, \{{\bf b}^{(i)}\}_{i\in [0, t)}, \{y_{i}\}_{i\in [0, t)}):\\
	{\sf FC.BVerify}(C, \{{\bf b}^{(i)}\}_{i\in [0, t)}, \{y_{i}\}_{i\in [0, t)}, \pi_y)=1
	\end{bmatrix}=1.
	$$
\end{definition}

Referring to \cite{WW23}, an FC scheme is {\em function binding} if no probabilistic polynomial time (PPT) adversary can open the  commitment $C$ to two distinct values for the same field vector. 
\begin{definition}
	[\bf Function binding] \label{df:FCSd}
	For any security parameter $\lambda$, any integer $n>0$, and any PPT adversary $\mathcal{A}$,
	$$
	\Pr
	\begin{bmatrix}
	{\sf pp} \leftarrow {\sf FC.Setup}(1^\lambda, 1^n),
	(C, \{{\bf b}^{(i)}, y_{i}, \hat{y}_{i}\}_{i\in [0, t)}, \pi_{y}, \pi_{\hat{y}}) \leftarrow\mathcal{A}({\sf pp}):\\
	({\sf FC.BVerify}(C, \{{\bf b}^{(i)}\}_{i\in [0, t)}, \{{y}_{i}\}_{i\in [0, t)}, \pi_{{y}})=1) ~\wedge\\
	({\sf FC.BVerify}(C, \{{\bf b}^{(i)}\}_{i=0}^{t-1}, \{\hat{y}_{i}\}_{i=0}^{t-1}, \pi_{\hat{y}})=1) \wedge(\exists j\in [0, t): y_j\neq \hat{y}_{j})\\
	\end{bmatrix}\leq {\sf negl}(\lambda).
	$$
\end{definition}

\begin{comment}	
Referring to \cite{KZG10}, the security of batch opening (called batch binding) requires that 
no adversary can batch open $\{f_{{\bf b}^{(i)}}\}_{i\in [t]}$ at the commitment $C$ in a manner that conflicts with the result of opening $f_{{\bf b}^{(j)}}(j\in [t])$ individually. 
\begin{definition}
	[\bf Batch binding]\label{df:batch1binding}
	For any security parameter $\lambda$, any integer $n>0$, and any PPT adversary $\mathcal{A}$,
	{\small
		$$
		\Pr
		\begin{bmatrix}
		{\sf pp} \leftarrow {\sf FC.Setup}(1^\lambda, 1^n),
		(C, \{{\bf b}^{(i)}, \hat{y}_{i}\}_{i\in [t]}, \pi_{\hat{y}}, j\in [t], y_{j}, \pi_{y_{j}}) \leftarrow\mathcal{A}({\sf pp}):\\
		{\sf FC.BVerify}(C, \{{\bf b}^{(i)}\}_{i\in [t]}, \{\hat{y}_{i}\}_{i\in [t]}, \pi_{\hat{y}})=1 ~\wedge\\
		{\sf FC.Verify}(C, {\bf b}^{(j)}, {y}_{j}, \pi_{{y}_{j}})=1 ~\wedge ~ (y_{j} \neq \hat{y}_{j}) \\
		\end{bmatrix}\leq {\sf negl}(\lambda).
		$$}
\end{definition}
\end{comment}

\subsection{Polynomial Commitments}
A {\em polynomial commitment} (PC) scheme \cite{KZG10,BGV11,PST13,WTS+18,ZXZ+20} enables a prover to commit to an $n$-variate polynomial of maximum degree $d$ (per variable), and later open its evaluation at any point by generating an evaluation proof. A PC scheme ${\sf PC}=({\sf Setup}, {\sf Commit}, {\sf Eval}, {\sf Verify})$ consists of four algorithms:
\begin{itemize}
	\it
	\item
	${\sf PC.Setup}(1^\lambda, 1^d, 1^n)\rightarrow {\sf pp}$:
	Given the security parameter $\lambda$, the maximum degree per variable $d$ and the number of variables $n$, outputs the public parameters ${\sf pp}$, an implicit input to all remaining algorithms.	
	\item
	${\sf PC.Commit}(f)\rightarrow C$:
	Outputs the commitment of $f$.
	\item
	${\sf PC.Eval}(f, {\bf r})\rightarrow (y, \pi)$:
	Generates a proof $\pi$ showing that $f({\bf r})=y$.
	\item
	${\sf PC.Verify}(C, {\bf r}, y, \pi)\rightarrow \{0, 1\}$: Returns 1 if for the committed polynomial $f$ it holds that $f({\bf r})=y$.
\end{itemize}

Informally, a PC scheme is {\em complete} if the verifier always accepts the proof for a correctly evaluated point.
\begin{definition}
	[\bf Completeness]\label{df:PCCom}
	A PC scheme is complete if for any $\lambda,n,d$,
	$$
	\Pr
	\begin{bmatrix}
	{\sf pp} \leftarrow {\sf PC.Setup}(1^\lambda, 1^d, 1^n),
	C \leftarrow {\sf PC.Commit}(f),\\
	(y, \pi) \leftarrow  {\sf PC.Eval}(f, {\bf r}):
	{\sf PC.Verify}(C, {\bf r}, y, \pi)=1
	\end{bmatrix}=1.
	$$
\end{definition}

A PC scheme is {\em knowledge sound} if for any PPT adversary that produces an accepting proof, there exists an extractor $\mathcal{E}_{PC}$ that extracts the committed polynomial $f$, such that the probability that $f({\bf r})\neq y$ is negligible.
\begin{definition}
	[\bf Knowledge Soundness]\label{df:PCKS}
	A PC scheme is knowledge sound, if for any $\lambda$, $n$, $d$ and PPT adversary $\mathcal{A}_{PC}$ there exists an extractor $\mathcal{E}_{PC}$ having access to $\mathcal{A}_{PC}$ such that:
	$$
	\Pr
	\begin{bmatrix}
	{\sf pp} \leftarrow {\sf PC.Setup}(1^\lambda, 1^d, 1^n),
	(C, y, \pi, {\bf r}) \leftarrow \mathcal{A}_{PC}({\sf pp}),\\
	f \leftarrow \mathcal{E}_{PC}^{{\cal A}_{PC}}({\sf pp}):
	({\sf PC.Verify}(C, {\bf r}, y, \pi)=1) \wedge\\
	((f({\bf r})\neq y )\vee (C\neq {\sf PC.Commit}(f)))
	\end{bmatrix}\leq {\sf negl}(\lambda).
	$$
\end{definition}
%A PC scheme is {\em zero-knowledge} if a verifier learns nothing more besides the evaluation. 

%A zero-knowledge polynomial commitment (ZK-PC) ensures an adversary cannot distinguish interactions with a real polynomial from those with a simulator (that does not know the polynomial), so no information about the polynomial itself is leaked.

%{\color{blue}Informally, a PC scheme is {\em zero-knowledge} if an adversary cannot distinguish interactions with a real polynomial from those with a simulator (that does not know the polynomial). For the formal definition, see Appendix \ref{app:def}.}

{\color{black}
For our construction, we use the PST multivariate polynomial commitment scheme \cite{PST13}, which has commitment and evaluation complexities {\em linear} in the size of the polynomial, and proof size and verification time logarithmic in it.}

\vspace{2mm}
\noindent
{\bf HyperEval Algorithm.}
In some scenarios, the prover needs to generate $2^n$ evaluation proofs of a multi-linear polynomial $f:\mathbb{F}_p^{n}\rightarrow \mathbb{F}_p$ at all hypercube points ${\bf i}\in \{0, 1\}^{n}$. This is formalized by the ${\sf PC.HyperEval}$ \cite{PPP25}:
\begin{itemize}
	\it
	\item
	${\sf PC.HyperEval}(f)\rightarrow \{y_i, \pi_i\}_{i\in [0, 2^n)}$: Generate $2^n$ evaluation proofs showing that $y_i=f({\sf Bin}(i))$ for all $i\in [0, 2^n)$.
\end{itemize}
The ${\sf HyperEval}$ algorithm for the PST scheme achieves ${\cal O}(n2^n)$ complexity \cite{SCP+22}.

\subsection{Vector Commitments}
A {\it vector commitment} (VC) scheme \cite{CF13,LY10} enables a prover to commit to a vector ${\bf m}\in \mathbb{F}_p^N$ and later prove that an element $m_i$ is the $i$-th element of the committed vector. 
A VC scheme ${\sf VC}=({\sf Setup}, {\sf Commit}, {\sf Open}, {\sf Verify})$ consists of the following algorithms:
\begin{itemize}
	\it
	\item
	${\sf VC.Setup}(1^\lambda, 1^{N})\rightarrow {\sf pp}$:
	Given the security parameter $\lambda$ and vector size $N$, outputs the public parameters, an implicit input to all remaining algorithms.
	\item
	${\sf VC.Commit}({\bf m})\rightarrow (C, {\sf aux})$:
	Outputs the vector commitment of ${\bf m}$ and an auxiliary information ${\sf aux}$.
	\item
	${\sf VC.Open}({\sf aux}, i, {\bf m})\rightarrow \pi_{i}$:
	Outputs an opening proof $\pi_{i}$ showing that ${m}_{i}$ is the $i$-th element of ${\bf m}$.
	\item
	${\sf VC.Verify}(C, i, {m}_{i}, \pi_{i})\rightarrow \{0, 1\}$: Returns 1 if ${m}_{i}$ is the $i$-th element of the committed vector.	
\end{itemize}

A VC scheme is {\em correct} if ${\sf VC.Verify}$ always outputs 1, provided that all algorithms are correctly executed. 
\begin{definition}
	[\bf Correctness]\label{df:VCCor}
	For any security parameter $\lambda$, any integer $N>0$, any vector ${\bf m}\in \mathbb{F}_p^N$ and any index $i\in [0, N)$, a VC scheme is correct if
	$$
	\Pr
	\begin{bmatrix}
	{\sf pp} \leftarrow {\sf VC.Setup}(1^\lambda, 1^N),
	(C, {\sf aux}) \leftarrow {\sf VC.Commit}({\bf m}),\\
	\pi_{i} \leftarrow  {\sf VC.Open}({\sf aux}, i, {\bf m}):
	{\sf VC.Verify}(C, i, {\bf m}[i], \pi_i)=1
	\end{bmatrix}=1.
	$$
\end{definition}

A VC scheme is {\it position binding} if the probability that a PPT adversary generates valid proofs for two different elements at the same index $i$ is negligible. 
\begin{definition}
	[\bf Position binding] \label{df:VCSd}
	For any security parameter $\lambda$, any integer $N>0$, and any PPT adversary $\mathcal{A}$,
	$$
	\Pr
	\begin{bmatrix}
	{\sf pp} \leftarrow {\sf VC.Setup}(1^\lambda, 1^{N}),
	({C}, i, m_i, m_i^\prime, \pi_{i}, \pi_{i}^\prime) \leftarrow\mathcal{A}({\sf pp}):\\
	({\sf VC.Verify}(C, i, m_i, \pi_{i})=1) \wedge \\
	({\sf VC.Verify}(C, i, m_i^\prime, \pi_{i}^\prime)=1) \wedge
	(m_i\neq m_i^\prime)  \\
	\end{bmatrix}\leq {\sf negl}(\lambda).
	$$
\end{definition}
%{\color{blue}A vector commitment scheme is {\em position hiding} \cite{PPP25} if an adversary cannot distinguish real openings from simulated ones, thereby hiding position information. The formal definition is given in Appendix~\ref{app:def}.}

%A vector commitment scheme is {\em position hiding} if an adversary cannot distinguish real openings of vector positions (generated from the actual vector) from simulated openings (produced without the vector, via a simulator), ensuring position information is kept secret. For the formal definition, see Appendix \ref{app:def}.

\vspace{2mm}
\noindent
{\bf OpenAll Algorithm.}
Recent VC schemes \cite{SCP+22,WUP23,TAB+20,LZ22,PPP25} support an $\sf OpenAll$ algorithm that {\it efficiently} generates all opening proofs in an offline pre-processing phase and later reply to open queries with no additional computation. The formal definition of the $\sf VC.OpenAll$ algorithm is as follows.
\begin{itemize}
	\item
	${\sf VC.OpenAll}({\sf aux}, {\bf m})\rightarrow \{\pi_{i}\}_{i\in [0, N)}$:
	Generate $N$ opening proofs, one for each element of the vector.
\end{itemize}

\section{A Functional Commitment with Batch Opening}\label{sec:ourFC}	
%1.1. 现有的论文提出了一个argument
%1.2. 这个argument是干什么的
B{\"u}nz et al. \cite{BMM+21} proposed a known-exponent {\it multi-exponentiation inner product argument} (MIPA) which {\color{black} employs structured-key variants of the commitment in \cite{AFG+16} to commit to $({\bf A},{\bf b})$ and}
allows one to prove that a value $U\in \mathbb{G}_1$ is equal to $\langle {\bf A}, {\bf b}\rangle$ against this commitment, where ${\bf A}\in \mathbb{G}_1^n$ is hidden and ${\bf b}\in \mathbb{F}_p^n$ is publicly known and structured of the form $(1, b, \ldots, b^{n-1})$ for some $b\in \mathbb{F}_p$.

%2.1 我们改进了什么
We observe that in the MIPA, the vector $\mathbf{b}$ can be any vector in $\mathbb{F}_p^n$ and extend the argument to the {\em first} FC scheme with batch openings for multi-exponentiations. 
%2.2 我们如何实现了这个改进
To develop the final FC scheme, we first transform the MIPA (which originally verified $\langle {\bf A}, {\bf b}\rangle$ for fixed ${\bf A}$ and ${\bf b}$) into an FC scheme that commits to ${\bf A}$, enabling verification of $\langle {\bf A}, {\bf b}\rangle$ for {\em any} ${\bf b}$.
Opening the commitment of vector $\bf A$ to $y_0=\langle {\bf A}, {\bf b}^{(0)}\rangle, \ldots, y_{t-1}=\langle {\bf A}, {\bf b}^{(t-1)} \rangle$ is equivalent to proving: 
\begin{equation}\label{eq:distinctB}
\{\langle {\bf A}, {\bf b}^{(i)}\rangle=y_i\}_{i\in [0, t)}.
\end{equation}
As shown in Lemma \ref{lem:equalityB} (Appendix \ref{app:B}), this task can be reduced to showing that a randomized linear combination of these equations holds:
\begin{equation}\label{eq:batchB}
\langle {\bf A}, \sum_{i\in [0, t)}r_i{\bf b}^{(i)}\rangle=\prod_{i\in [0, t)} y_i^{r_i},
\end{equation}
where each $r_i$ is random and may be computed with a hash function. 
%2.3 我们方案可以用来构造什么
We then show that this FC scheme gives a VC scheme in which ${\sf VC.OpenAll}$ algorithm runs in linear time (Section \ref{sec:vc}).
Below we give a formal description of our FC scheme, assuming without loss of generality that the vector length $n$ is a power of 2, i.e., $n = 2^\ell$ for some integer $\ell$.
\begin{itemize}
	\it
	\item
	${\sf FC.Setup}(1^\lambda, 1^n)$:
	Choose $(p, \mathbb{G}_1, \mathbb{G}_2, \mathbb{G}_{\rm T}, e, g_1, g_2)\leftarrow {\sf BG}(1^\lambda)$ and $\beta\leftarrow \mathbb{F}_p$. Compute
	$
	{\bf v}=\big((g_2)^{\beta^{2i}}\big)_{i\in [0, n)} 
	$
	and output  ${\sf pp}={\bf v}$.
	{\color{black} Note that $\beta$ must never be known to the adversary, and thus our scheme requires a trusted setup.}
	
	\item
	${\sf FC.Commit}({\bf A})$:
	Given {\color{black}the implicit input ${\sf pp} = {\bf v} \in \mathbb{G}_2^n$} and a vector ${\bf A}\in \mathbb{G}_1^n$, output  a commitment 
	$
	{C}={\bf A}*{\bf v}.
	$
	
	\item
	${\sf FC.BOpen}(C, {\bf A}, \{{\bf b}^{(i)}\}_{i\in [0, t)}, \{y_{i}\}_{i\in [0, t)})$:
	Parse {\color{black} the implicit input} ${\sf pp}$ as ${\bf v}$.
	Set
	\begin{equation}\label{eqn:b}
	{\bf b}=\sum_{i\in [0, t)}r_i{\bf b}^{(i)}, 
	\end{equation}	
	where $r_i=H({\bf b}^{(i)}, C, \{{\bf b}^{(k)}\}_{k\in [0, t)}, \{y_{k}\}_{k\in [0, t)}),\forall i\in[0, t),$
	and $H: \{0, 1\}^*\rightarrow \mathbb{F}_p$ is a hash function. 		  
	Set ${\bf A}_0={\bf A}, {\bf v}_{0}={\bf v}, {\bf b}_{0}={\bf b}$, and let $x_0$ be the empty string.
	For $j=1, 2, \ldots, \ell$, compute 
	\begin{equation}\label{eq:folding}
	\begin{aligned}
	{\bf L}_j&=(({\bf A}_{j-1})_R*({\bf v}_{j-1})_L, \langle ({\bf A}_{j-1})_R, {({\bf b}_{j-1})_L}\rangle),\\
	{\bf R}_j&=(({\bf A}_{j-1})_L*({\bf v}_{j-1})_R, \langle({\bf A}_{j-1})_L, {({\bf b}_{j-1})_R}\rangle),\\
	x_j&=H^\prime(x_{j-1}, {\bf L}_j, {\bf R}_j), \\
	\hspace{2mm}{\bf A}_{j}&=({\bf A}_{j-1})_L\circ(({\bf A}_{j-1})_R)^{x_j},\\
	{\bf v}_{j}&=({\bf v}_{j-1})_L\circ(({\bf v}_{j-1})_R)^{1/x_j}, \\
	\hspace{2mm}{\bf b}_{j}&=({\bf b}_{j-1})_L+(1/x_j)\cdot({\bf b}_{j-1})_R,
	\end{aligned}
	\end{equation}
	where $H^\prime:\{0,1\}^*\rightarrow \mathbb{F}_p$ is a hash function. 
	Output   	
	$\pi_y=\big\{\{{\bf L}_j, {\bf R}_j\}_{j=1}^{\ell}, {\bf A}_\ell\big\}.$
	
	\item
	${\sf FC.BVerify}(C, \{{\bf b}^{(i)}\}_{i\in [0, t)}, \{y_{i}\}_{i\in [0, t)}, \pi_y)$:
	{\color{black}Parse the implicit input ${\sf pp}$ as ${\bf v}$.}
	Compute $\{r_i\}_{i\in[0, t)}$ and $\bf b$ as in algorithm $\sf FC.BOpen$, and then compute 
	\begin{equation}\label{eq:comy}
		y=\prod_{i\in [0, t)} y_{i}^{r_i}.
	\end{equation}
	Let ${\bf C}_0=({C}, y)$ and  let $x_0$ be the empty string.
	Parse $\pi_y$ as $\{\{{\bf L}_j, {\bf R}_j\}_{j=1}^{\ell}, {\bf A}_\ell\}$.
	For $j=1,2,\ldots,\ell$, compute $x_j,{\bf v}_j$ and ${\bf b}_j$ as per Eq. \eqref{eq:folding}, and
	compute
	$${\bf C}_j = {\bf L}_j^{x_j}\circ {\bf C}_{j-1} \circ ({\bf R}_j)^{1/x_j}.$$
	Finally, if 
	${\bf C}_\ell=({\bf A}_\ell*{\bf v}_\ell, \langle {\bf A}_\ell, {{\bf b}_\ell}\rangle),$
	output 1; otherwise, output 0.
\end{itemize}

\vspace{2mm}
\noindent
{\bf Correctness and security.}	
The correctness of our FC scheme is proven below, while its security (function binding) is proven in Appendix \ref{sec:FCSnd}.

\begin{theorem}
	The FC scheme satisfies the correctness property (Definition  \ref{df:FCCor}).
\end{theorem}
\begin{proof}
	To prove the correctness of the FC scheme, it suffices to show that for all $k\in \{0, 1, \ldots, \ell\}$, ${\bf C}_i=({\bf A}_i*{\bf v}_i, \langle {\bf A}_i, {{\bf b}_i}\rangle)$.
	We begin by proving this statement for the base case 	$k=0$.
	Since for each $i\in [0, t)$, $y_i = \langle \mathbf{A}, \mathbf{b}^{(i)} \rangle$ and $r_i$ is computed as in algorithm $\sf FC.BOpen$, we obtain
	$$
	\prod_{i\in [0, t)} y_i^{r_i}
	= \prod_{i\in [0, t)} \langle \mathbf{A}, \mathbf{b}^{(i)} \rangle^{r_i}
	= \left\langle \mathbf{A}, \sum_{i\in [0, t)} r_i\mathbf{b}^{(i)} \right\rangle.
	$$
	By equations $\eqref{eqn:b}$ and $\eqref{eq:comy}$, it follows that
	$y = \langle \mathbf{A}, \mathbf{b} \rangle.$
	Hence, it is evident that
	${\bf C}_0=({C}, y)=({\bf A}_0*{\bf v}_0, \langle {\bf A}_0, {{\bf b}_0}\rangle),$
	as required.
		
	By mathematical induction, it remains to show   
	the statement is true for $k=j$ when it is true for $k=j-1$. 
	For simplicity, we denote $\bar{{\bf A}}={\bf A}_{j-1},  \bar{\bf v}={\bf v}_{j-1}, $ and $\bar{\bf b}={\bf b}_{j-1}$.
	According to the algorithm ${\sf FC.BVerify}$, ${\bf C}_{j}={\bf L}_{j}^{x_{j}}\circ {\bf C}_{j-1} \circ ({\bf R}_{j})^{1/x_j}$. 
	By Eq. \eqref{eq:folding} and the induction hypothesis, 
	{\small
		\begin{align*}
		{\bf C}_{j}&=(\bar{{\bf A}}_R*\bar{{\bf v}}_L, \langle \bar{{\bf A}}_R, {\bar{{\bf b}}_L}\rangle)^{x_{j}} 
		\circ (\bar{{\bf A}}*\bar{\bf v}, \langle\bar{{\bf A}}, \bar{\bf b}\rangle)\circ (\bar{{\bf A}}_L*\bar{\bf v}_R, \langle\bar{{\bf A}}_L, {\bar{\bf b}_R}\rangle)^{1/x_j} \\
		&=
		((\bar{{\bf A}}_R^{x_{j}}*\bar{\bf v}_L)\cdot (\bar{{\bf A}}*\bar{\bf v}) \cdot (\bar{{\bf A}}_L^{1/x_j}*\bar{\bf v}_R), 
		\langle \bar{{\bf A}}_R^{x_{j}}, {\bar{\bf b}_L}\rangle \cdot \langle \bar{{\bf A}}, \bar{\bf b}\rangle \cdot \langle\bar{{\bf A}}_L^{1/x_j}, {\bar{\bf b}_R}\rangle).
		\end{align*}}\ignorespaces
	Since $\bar{{\bf A}}*\bar{\bf v}=(\bar{{\bf A}}_L*\bar{\bf v}_L) \cdot (\bar{{\bf A}}_R*\bar{\bf v}_R)$ and $\langle \bar{{\bf A}}, \bar{\bf b}\rangle=\langle \bar{{\bf A}}_L, \bar{\bf b}_L\rangle \cdot \langle \bar{{\bf A}}_R, \bar{\bf b}_R\rangle$, we have 
	%merge terms by $\bar{{\bf v}}_L$, $\bar{\bf v}_R$, $\bar{\bf b}_L$ and $\bar{\bf b}_R$,
	{\small
		$${\bf C}_{j}=(((\bar{{\bf A}}_R^{x_{j}} \circ \bar{{\bf A}}_L)*\bar{\bf v}_L) \cdot ((\bar{{\bf A}}_R \circ \bar{{\bf A}}_L^{1/x_j})*\bar{\bf v}_R), 
		\langle \bar{{\bf A}}_R^{x_{j}}\circ\bar{{\bf A}}_L, {\bar{\bf b}_L}\rangle \cdot \langle \bar{{\bf A}}_R\circ\bar{{\bf A}}_L^{1/x_j}, \bar{\bf b}_R\rangle).$$}\ignorespaces
	Since $\bar{{\bf A}}_R \circ \bar{{\bf A}}_L^{1/x_j}=
	(\bar{{\bf A}}_R^{x_{j}} \circ \bar{{\bf A}}_L)^{1/x_j}$, by replacing the left-hand side of this equality with the right-hand side in ${\bf C}_{j}$, we have	
	$${\bf C}_{j}=
	((\bar{{\bf A}}_R^{x_{j}} \circ \bar{{\bf A}}_L)*(\bar{\bf v}_L\circ \bar{\bf v}_R^{1/x_j}), \langle \bar{{\bf A}}_R^{x_{j}}\circ\bar{{\bf A}}_L, {\bar{\bf b}_L}+{1/x_j} \cdot \bar{\bf b}_R\rangle).	
	$$
	According to Eq. \eqref{eq:folding}, ${\bf C}_{j}=({\bf A}_{j}*{\bf v}_{j}, \langle {\bf A}_{j}, {{\bf b}_{j}}\rangle).$
	\qed	
\end{proof}

\vspace{2mm}
\noindent
{\bf Faster verification.}
The FC scheme uses the same approach as the known-exponent MIPA in B{\"u}nz et al. \cite{BMM+21}, 
which allows the verifier to outsource the computation of ${\bf v}_\ell$ to the untrusted prover and reduces verification time.
However, this comes at the cost of additionally relying on the $n$-SDH assumption ({\bf Assumption} \ref{asp:qSDH}). 
We implicitly assume the faster verifier.

\vspace{2mm}
\noindent
{\bf Efficiency.}
The efficiency of the FC scheme with faster verification is analyzed below and summarized in Table~\ref{tb:fcEff}. 
In our analysis of group operations, we focus on pairings and exponentiations, ignoring the cheaper multiplications.

\begin{itemize}
	%\item
	%{\em Public parameter.}
	%The original public parameter consists of $n$ elements in $\mathbb{G}_2$. 
	%To enable faster verification, $n$ additional elements in $\mathbb{G}_2$ were included.
	
	\item
	{\em Proof size.}
	The original proof contains $2\log n$ elements in $\mathbb{G}_{\rm T}$ and $2\log n+1$ elements in $\mathbb{G}_1$.
	To enable faster verification, two additional elements in $\mathbb{G}_2$ (i.e., ${\bf v}_{\ell}$ and its proof) were included.	
	Hence, the proof size is $\mathcal{O}(\log n)$.
	
	\item
	%Generating the public parameter requires $2n$ exponentiations in $\mathbb{G}_2$.
	{\em Commit.}
	Committing to an $n$-sized vector requires $n$ pairings. 
	
	\item
	{\em Batch opening.}
	Computing vector ${\bf b}$ requires $tn$ field operations.
	Computing $\{{\bf L}_j, {\bf R}_j\}_{j=1}^{\ell}$ requires $2n$ pairings and $2n$ exponentiations in $\mathbb{G}_1$.
	Rescaling ${\bf A}, {\bf v}, {\bf b}$ requires $n$ exponentiations in $\mathbb{G}_1$ and $\mathbb{G}_2$, $n$ field operations.	
	Moreover, to achieve faster verification, the proof of ${\bf v}_{\ell}$ must be computed, and this computation requires performing $2n$ exponentiations in $\mathbb{G}_2$.
	
	\item
	{\em Batch verification.}
	Computing ${\bf b}$ requires $tn$ field operations, and computing $y$ requires $t$ exponentiations in $\mathbb{G}_1$.
	Rescaling the commitments and ${\bf b}$ requires $2\log n$ exponentiations in $\mathbb{G}_{\rm T}$ and $\mathbb{G}_1$, and $n$ field operations.
	The final check equation requires 1 pairing and an exponentiation in $\mathbb{G}_1$.
	With faster verification, the verifier no longer needs to compute ${\bf v}_l$, but needs to verify its correctness; this requires two pairings and $\log n$ field operations.
\end{itemize}

\begin{table}[htbp]
	\centering
	\caption{The efficiency of our FC scheme with faster verification.
	}
	\begin{tabular}{c|c c c}
		\hline
		{\bf Algorithms}&  \quad{\bf  Field operations } \quad  &   {\bf \quad Exponentiations \quad}  &  {\bf \quad Pairings \quad}  \\
		\hline
		${\sf FC.Commit}$&-&-&$\mathcal{O}(n)$\\
		%\hline
		%${\sf FC.Open}$ & $\mathcal{O}(n)$&$\mathcal{O}(n)$&$\mathcal{O}(n)$\\
		%\hline
		%${\sf FC.Verify}$ &$\mathcal{O}(n)$&$\mathcal{O}(\log n)$&$\mathcal{O}(1)$\\
		\hline
		${\sf FC.BOpen}$&$\mathcal{O}(tn)$ & $\mathcal{O}(n)$&$\mathcal{O}(n)$\\
		\hline
		${\sf FC.BVerify}$&$\mathcal{O}(tn)$&$\mathcal{O}(\log n+t)$&$\mathcal{O}(1)$\\
		\hline
	\end{tabular}
	\label{tb:fcEff}
\end{table}	

\section{A VC Scheme with Linear-Time OpenAll}\label{sec:vc}

{\color{black}
	Based on the   scheme $\sf FC$ from Section \ref{sec:ourFC}, we construct FlexProofs (denoted by $\sf VC$), a VC scheme that achieves linear-time ${\sf VC.OpenAll}$ with a flexible batch size parameter $b$, where  $b$ provides a tradeoff 
	and the running time of ${\sf VC.OpenAll}$ can be continuously reduced by increasing $b$.
	Our scheme is directly compatible with the zkSNARKs based on multi-linear polynomials. 
	
	\subsection{FlexProofs}
	%Introduce VC.Setup
	FlexProofs is structured in two layers. We set $\mu = \sqrt{N}$ and partition the vector $\mathbf{m} \in \mathbb{F}_p^N$ into $\mu$ subvectors of length $\mu$:
	$$
	\mathbf{m} = \bigl(\mathbf{m}_0, \mathbf{m}_1, \dots, \mathbf{m}_{\mu-1}\bigr).
	$$
	In the first layer, each $\mathbf{m}_j$ ($j \in [0,\mu)$) is encoded using its multilinear extension 
	$f_j : \mathbb{F}_p^{\log \mu} \to \mathbb{F}_p$
	and committed via a   commitment scheme $\sf PC$ \cite{PST13}  for multilinear polynomials. In the second layer, the vector of the commitments of all ${\bf m}_j(j\in [0,\mu))$ is committed using the scheme $\sf FC$ from Section \ref{sec:ourFC}.
	Specifically, the algorithms ${\sf VC.Setup}$ and ${\sf VC.Commit}$ in our VC scheme  are as follows.
	\begin{itemize}
		\it
		\item
		${\sf VC.Setup}(1^\lambda, 1^{N})$:
		Generate public parameters for  
		$\sf PC$ and $\sf FC$ by invkoing
		\begin{equation}\label{eq:vcpp}
		{\sf pp}_{PC} \leftarrow {\sf PC.Setup}(1^\lambda, 1, \log \mu), 
		~{\sf pp}_{FC}\leftarrow {\sf FC.Setup}(1^\lambda, 1^\mu).
		\end{equation}
		Output  the public parameters ${\sf pp}=\{{\sf pp}_{PC}, {\sf pp}_{FC}\}$.
		\item
		${\sf VC.Commit}({\bf m})$:
		For each $j\in [0, \mu)$, the subvector ${\bf m}_j$ is encoded using its multilinear extension $f_j$ and then committed as
		\begin{equation}\label{eq:genCi}
		C_j\leftarrow {\sf PC.Commit}(f_j).
		\end{equation}
		For the vector ${\bf C}=(C_0, \ldots, C_{\mu-1})$ of $\mu$ polynomial commitments, generate a functional commitment
		\begin{equation}\label{eq:genC}
		C\leftarrow {\sf FC.Commit}(\bf C).
		\end{equation}
		Output the commitment $C$ and an auxiliary information string ${\sf aux}={\bf C}$.
	\end{itemize}

	Given the commitment process defined by \eqref{eq:genCi} and \eqref{eq:genC}, there is a {\em two-step} method of proving the correctness of 
	${\bf m}[i]$ for each $i\in [0, N)$.
	First, use our FC scheme to prove the correctness of $C_{\lfloor i/\mu\rfloor}$ with respect to $C$, i.e.,
	$C_{\lfloor i/\mu\rfloor}=\langle {\bf C}, {\bf u}_{\lfloor i/\mu\rfloor}\rangle$.
	Second, use the PC scheme to prove the correctness of ${\bf m}[i]$ with respect to $C_{\lfloor i/\mu\rfloor}$, i.e.,
	$f_{\lfloor i/\mu\rfloor}({\sf Bin}(i\bmod \mu))={\bf m}[i]$.
	Since both   steps have time complexity ${\cal O}(\sqrt{N})$, the computational cost of
	the two-step method  is ${\cal O}(\sqrt{N})$. 
	Consequently, proving the correctness of all elements of ${\bf m}$ one after the other 
	with the two-step method requires ${\cal O}(N\sqrt{N})$ time.
	By examining the two-step proving process, we observe that: (1) proving the correctness of $C_j$ with respect to $C$ takes place $ \mu$ times and   can be done only once for all elements in  
	${\bf m}_j$; (2) for all $j\in[0,\mu)$, the $\mu$ processes of proving the correctness
	of all elements in ${\bf m}_j$ with respect to $C_j$ share a similar structure and thus can be combined. 
	%by invoking $\sf PC.HyperEval$ on $f_j$.
	Based on these observations, we propose an {\em OpenAll} protocol, in which the prover can prove the correctness of all elements in ${\bf m}$ in ${\cal O}(N)$ time.
	
	%When proving the correctness of all elements in the vector $\mathbf{m}$, we observe that (1) the proving process for all elements in ${\bf m}_j$ includes proving the correctness of $C_j$ with respect to $C$, and (2) the process of proving the correctness of all elements in ${\bf m}_j$ with respect to $C_j$ follows a same pattern for all $j\in[0, \mu)$. 

	%It is clear that for each $i\in [0, N)$, proving the correctness of $m_i$ (in ${\bf m}_{\lfloor i/\mu \rfloor}$) requires two steps. First, use our FC scheme to prove the correctness of $C_{\lfloor i/\mu \rfloor}$ with respect to $C$; then, use the PC scheme to prove the correctness of $m_i$ with respect to $C_{\lfloor i/\mu \rfloor}$.

	\subsection{Linear-Time OpenAll}
	
	In this section,  we first present {\em OpenAll} as an interactive protocol and then make it non-interactive
	with the Fiat-Shamir variant of \cite{ZXH+22}, yielding the algorithms  $\sf VC.OpenAll$ and $\sf VC.Verify$.
	In the non-interactive version of our VC scheme, the executions of the algorithms   $\sf VC.Commit$ and   $\sf VC.OpenAll$ form a 
	{\em pre-processing} phase and the latter algorithm generates 
	$N$ opening proofs  for any vector $\bf m$ of length $N$. In particular, the
	$N$ 
	opening proofs can be used to directly
	reply any query to  $\sf VC.Open$, which is an algorithm  different from $\sf VC.OpenAll$ 
	and every time generates  an opening proof  for any  {\em single} element of the vector.

	The interactive protocol  {\it OpenAll} is a protocol between a prover holding a vector ${\bf m}$ and $N$ verifiers $\{\mathcal{V}_i\}_{i\in [0, N)}$, where each verifier $\mathcal{V}_{i}$ holds a commitment $C$
	(to $\bf m$) and a value ${m}_{i}$, and verifies that ${m}_{i}$ is the $i$-th element of the vector committed in ${C}$ by interacting with the prover.

	\vspace{2mm}
	\noindent
	{\bf Intuition of our protocol.}
	Our {\it OpenAll} protocol consists of two steps: (1) proving the correctness of $C_j$ with respect to $C$ for all $j \in [0,\mu)$, and (2) combining the $\mu$ processes of proving the correctness of all elements in ${\bf m}_j$ with respect to $C_j$.
	For the first step, if 
	the prover uses the scheme $\sf FC$ from Section \ref{sec:ourFC} to generate a proof for each $C_j$,
	then it needs to   generate $\mu$ proofs and thus performs  $\mathcal{O}(N)$ field operations and $\mathcal{O}(N)$ cryptographic operations (e.g., group exponentiations or pairings) according to Table \ref{tb:fcEff}.
	To reduce this cost, the prover partitions the vector ${\bf C}$ into blocks and uses $\sf FC$ to generate one batch proof per block.
	With a batch size of $b$, the prover only needs to generate $\lceil\mu/b\rceil$ batch proofs, which require $\mathcal{O}(N)$ field operations and $\mathcal{O}(N/b)$ cryptographic operations.
	For the second step, we combine the $\mu$ processes by folding $\mu$ polynomials into a single polynomial \cite{GWC19,BDF+21},
	called {\em folded}  polynomial, which 
	is essentially a linear combination of the 
	$\mu$   polynomials, such that if an evaluation claim at a hypercube point is wrong in one of the  $\mu$ polynomials, then it will be wrong for the folded polynomial with overwhelming probability. For the folded polynomial, we invoke $\sf PC.HyperEval$. 
	
	\begin{figure*}[htbp!]
		\begin{boxedminipage}{\textwidth}
			\captionsetup{labelfont=bf, labelformat=simple, labelsep=colon, name=Protocol}
			\centering
			\caption{{\bf Efficient OpenAll}. We assume a prover holding a vector ${\bf m}\in \mathbb{F}_p^N$, the commitment ${C}$ and an auxiliary information string ${\sf aux}(={\bf C})$. 
				We assume $N$ verifiers $\{\mathcal{V}_{i}\}_{i\in [0, N)}$, where 	 
				each verifier $\mathcal{V}_{i}$ holds the commitment ${C}$ and an element ${m}_{i}$ and wants to verify whether ${m}_{i}$ is the $i$-th element of ${\bf m}$.}		
			\begin{itemize}[leftmargin=*, label=\textbullet]	
				\item
				{\bf Step 1: Prove the correctness of ${\bf C}_k$ for all $k\in [0, \lceil\mu/b\rceil)$.}
				
				\begin{enumerate}[leftmargin=*, itemsep=0pt, topsep=0pt]
					\item[1)]
					For each $k\in [0, \lceil\mu/b \rceil)$,  the prover sets $S_k=[kb, (k+1)b)$,  generates a batch proof 
					\begin{equation}\label{eq:genFCBatchProof}
					\pi_{{\bf C}_k}\leftarrow {\sf FC.BOpen}({C}, {\bf C}, \{{\bf u}_a\}_{a\in S_k}, \{C_a\}_{a\in S_k}),
					\end{equation}
					and sends  $(\mathbf{C}_k,\pi_{{\bf C}_k})$ to the verifiers
					$\{\mathcal{V}_{i}\}_{i \in [k \mu b,(k+1)\mu b)}.$	
					\item[2)]
					The verifier~$\mathcal V_i$ sets $k = \lfloor i / (\mu b) \rfloor$ and $S_k=[kb, (k+1)b)$, and checks the correctness of $\mathbf{C}_k$ by executing
					$${\sf FC.BVerify}({C}, \{{\bf u}_a\}_{a\in S_k}, \{C_a\}_{a\in S_k}, \pi_{{\bf C}_{k}}).$$ 	
				\end{enumerate}	
				
				\item {\bf Step 2: Fold all $\mu$ polynomials $\{f_j\}_{j=0}^{\mu-1}$ into a single polynomial $g^*$ and then invoke $\sf PC.HyperEval$ on the folded polynomial $g^*$.}
				\begin{enumerate}[leftmargin=*, itemsep=0pt, topsep=0pt]
					\item[1)]
					The prover receives random values $\{r_j\}_{j\in[0,\mu)}$ from the verifiers and  computes
					\begin{equation}\label{eq:ProverRand}
					g_j = r_j f_j, \hspace{5mm} 
					D_j = C_j^{r_j}.
					\end{equation}
					For each $i \in [0,N)$, the verifier $\mathcal{V}_i$, holding $m_i$ and $C_{\lfloor i/\mu \rfloor}$, computes
					\begin{equation}\label{eq:VerifierRand}
					{y}_{\lfloor i/\mu \rfloor, (i\bmod \mu)} = r_{\lfloor i/\mu \rfloor}\,m_i,
					\hspace{5mm}
					D_{\lfloor i/\mu \rfloor} = C_{\lfloor i/\mu \rfloor}^{r_{\lfloor i/\mu \rfloor}}
					\end{equation}
					The prover folds the randomized polynomials $\{g_j\}_{j \in [0,\mu)}$ into a single polynomial in a tree-like, bottom-up fashion. At the leaf level, the state of the $j$-th leaf is initialized as
					$\big({\cal P}: \{g_j\},\; \vec{\cal V}_j: \{(y_{j,a}, D_j)\}_{a \in [0,\mu)}\big),$
					where the prover ${\cal P}$ holds only the polynomial $g_j$, and 
					$\vec{\cal V}_j = (\mathcal{V}_{j\mu}, \ldots, \mathcal{V}_{(j+1)\mu-1})$ is the vector of verifiers, with each $\vec{\cal V}_j[a]$ holding ${y}_{j, a}$ claimed to be $g_j({\sf Bin}(a))$ together with the commitment $D_j$.	
					For each internal node $w$, the state
					$
					{\cal S}_w = ({\cal P}: \{g_w\},\; \vec{\cal V}_w: \{(y_{w,a}, D_w)\}_{a \in [0,\mu)})
					$
					is computed from its left child ${\cal S}_{w_l}$ and right child ${\cal S}_{w_r}$ using the {\bf Sub-protocol}.
					This process is repeated recursively up the tree until the root.	
					
					\item[2)]				
					At the root, the final state is 
					$
					{\cal S}^* =
					\big({\cal P}: \{g^*\},\;
					\vec{\cal V}^*: \{(y^*_a, D^*)\}_{a \in [0,\mu)}\big).
					$
					The prover follows the ${\sf PC.HyperEval}$ algorithm to generate $\mu$ proofs $\{\pi^*_a\}_{a\in [0, \mu)}$ proving that each $y_{a}^{*}=g^*({\sf Bin}(a))$, and sends $\pi^*_a$ to verifiers $\{\mathcal{V}_{j\mu+a}\}_{j\in [0, \mu)}$.
					Verifiers $\{\mathcal{V}_{j\mu+a}\}_{j\in [0, \mu)}$ validates by running ${\sf PC.Verify}(D^*, {\sf Bin}(a), y_a^*, \pi^*_a)$.
				\end{enumerate}
			\end{itemize}
		\end{boxedminipage}
	\end{figure*}
	
	\vspace{2mm}
	\noindent
	{\bf Protocol description.}
	Formally   our interactive {\it OpenAll} protocol is shown in {\bf Protocol 1} and    consists of the following two steps.}

\begin{itemize}
	\item
	\textit{\underline{The first step:}}
	The prover partitions the vector ${\bf C}$ into $\lceil \mu / b \rceil$ blocks of size $b$ such that
	$$
	\mathbf{C} = (\mathbf{C}_0, \dots, \mathbf{C}_{\lceil \mu / b \rceil - 1}),
	$$
	where $\mathbf{C}_i = (C_{i b}, \dots, C_{(i+1)  b - 1})$.
	For each $k\in [0, \lceil\mu/b \rceil)$, the prover invokes $\mathsf{FC.BOpen}$ to generate a batch proof for the correctness of $\mathbf{C}_k$, and sends $\mathbf{C}_k$ together with this proof to the verifiers $\{\mathcal{V}_{i}\}_{i \in [k b\mu,(k+1) b \mu)}$.
	These verifiers then validate using $\mathsf{FC.BVerify}$.
	%Therefore, the larger the batch size, the shorter the prover time, but the longer each verifier's proof and verification time.
	%To balance these metrics, we choose $b=\log N$.
	\item
	{\color{black}
		\textit{\underline{The second step:}}
		At this point,   ${\cal V}_i$ assures that $C_{\lfloor i/\mu \rfloor}$ is the ${\lfloor i/\mu \rfloor}$-th element of the vector committed in $C$. It remains to verify that $f_{\lfloor i/\mu \rfloor}({\sf Bin}(i\bmod \mu))={\bf m}[i]$.
		The prover has to prove this identity for all $i\in[0,N)$ or equivalently generate proofs for the evaluations of the $\mu$ polynomials $\{f_j\}_{j=0}^{\mu-1}$ at all boolean hypercube points in $\{0,1\}^{\log \mu}$.
		If the prover  invokes $\sf PC.HyperEval$ on the $\mu$ polynomials one by one, the total time complexity will be ${\cal O}(N \log \sqrt{N})$.
		To be more efficient, we first fold the $\mu$ polynomials  into a single polynomial $g^*$ and then invoke $\sf PC.HyperEval$ on the folded polynomial $g^*$.

		%The prover performs the folding in a tree-structured fashion. 
		%In this way, the prover's workload and each verifier's costs are balanced: the prover’s computation cost is $\mathcal{O}(N)$, while each verifier’s communication and computation cost is $\mathcal{O}(\log \mu)$.
		
		\begin{comment}
		We now briefly discuss how to implement this folding scheme with lower computational complexity and communication complexity.
		At the end of the fold scheme, the verifier $\mathcal{V}_{i}$ should hold
		$$y^*_j = \sum_{k \in [0, \mu)} r_k{m}_{k\mu+j}
		\quad\text{and}\quad
		C^* = \prod_{k \in [0, \mu)} C_k^{r_k},$$
		where $j=i \mod \mu$.
		If the prover sends the collection \(\{m_{k\mu + j}, C_k\}_{k \in [0,\mu) \setminus \{\lfloor i/\mu \rfloor\}}\) directly to each $\mathcal{V}_{i}$, then the prover’s computation cost is $\mathcal{O}(1)$, while each verifier’s communication and computation cost grows to $\mathcal{O}(\mu)$.
		If instead the prover sends the values
		$\sum_{k \in [0,\mu) \setminus \{\lfloor i/\mu \rfloor\}} r_k\,m_{k\mu + j}$
		{and}
		$\prod_{k \in [0,\mu) \setminus \{\lfloor i/\mu \rfloor\}} C_k^{r_k}$
		then the prover’s computation rises to $\mathcal{O}(N\mu)$, but each verifier only incurs constant $\mathcal{O}(1)$ communication and computation cost.
		To strike a balance between the prover's workload and each verifier's costs, the prover performs the folding in a tree‑structured fashion.
		\end{comment}
		Next, for clarity, each verifier step is presented immediately after its corresponding prover step to illustrate all participants’ actions.
		To efficiently compute the folded polynomial $g^*$, the verifiers send
		a set $\{r_j\}_{j\in[0,\mu)}$ of random values to the prover. The prover then uses each  
		random value $r_j$  to compute a randomized polynomial  $g_j = r_j f_j$ and a commitment 
		$D_j = C_j^{r_j}$. Meanwhile, for each $i \in [0,N)$, the verifier $\mathcal{V}_i$, holding $m_i$ and $C_{\lfloor i/\mu \rfloor}$, computes ${y}_{\lfloor i/\mu \rfloor, (i\bmod \mu)} = r_{\lfloor i/\mu \rfloor}\, m_i$ and $D_{\lfloor i/\mu \rfloor} = C_{\lfloor i/\mu \rfloor}^{r_{\lfloor i/\mu \rfloor}}$.
		The prover folds the randomized polynomials $\{g_j\}_{j \in [0,\mu)}$ into a single polynomial in a tree-like, bottom-up fashion. At the leaf level, the state of the $j$-th leaf is initialized as
		$\big({\cal P}: \{g_j\},\; \vec{\cal V}_j: \{(y_{j,a}, D_j)\}_{a \in [0,\mu)}\big),$
		where the prover ${\cal P}$ holds only the polynomial $g_j$, and 
		$\vec{\cal V}_j = (\mathcal{V}_{j\mu}, \ldots, \mathcal{V}_{(j+1)\mu-1})$ is the vector of verifiers, with each $\vec{\cal V}_j[a]$ holding ${y}_{j, a}$ claimed to be $g_j({\sf Bin}(a))$ together with the commitment $D_j$.	
		For each internal node $w$, the state
		$$
		{\cal S}_w = ({\cal P}: \{g_w\},\; \vec{\cal V}_w: \{(y_{w,a}, D_w)\}_{a \in [0,\mu)})
		$$
		is computed from its left child ${\cal S}_{w_l}$ and right child ${\cal S}_{w_r}$ using the {\bf Sub-protocol}.
		This process is repeated recursively up the tree until the root.	}
	
	\vspace{-4mm}
	\begin{figure*}[htbp!]
		\centering
		\begin{boxedminipage}{\textwidth}
			\captionsetup{labelfont=bf, labelformat=simple, labelsep=colon, name=Protocol}
			{\bf Sub-protocol: Computing the parent state from the two child states.} 
			%given two child states ${\cal S}_{w_0}, {\cal S}_{w_1}$, compute the parent state ${\cal S}_v$. 
			For each $z \in \{l,r\}$, the child node $w_z$ is associated with a state  
			$$
			{\cal S}_{w_z} = ({\cal P}: \{g_{w_z}\},\; \vec{\cal V}_{w_z} : \{(y_{w_z,a}, D_{w_z})\}_{a \in [0,\mu)}),
			$$
			where the prover ${\cal P}$ holds the polynomial $g_{w_z}$, and for each $a \in [0,\mu)$, the verifiers $\vec{\cal V}_{w_z}[a]$ hold a point $y_{w_z, a}$ claimed to be $g_{w_z}({\sf Bin}(a))$
			together with the commitment $D_{w_z}$.  
			At the end of the folding, the parent node $w$ has a state  
			$$
			{\cal S}_w = ({\cal P}: \{g_w\},\; \vec{\cal V}_w : \{(y_{w,a}, D_w)\}_{a \in [0,\mu)}),
			$$
			where the prover holds the polynomial $g_w$, and  
			$
			\vec{\cal V}_w = \bigl(\vec{\cal V}_{w_l}[0] \cup \vec{\cal V}_{w_r}[0],\; \dots,\; \vec{\cal V}_{w_l}[\mu-1] \cup \vec{\cal V}_{w_r}[\mu-1]\bigr)
			$
			collects the verifier vectors. For each $a \in [0,\mu)$, both $\vec{\cal V}_{w_l}[a]$ and $\vec{\cal V}_{w_r}[a]$ now hold the same claimed evaluation $y_{w,a}$ and the commitment $D_w$.  		
			Our protocol works as follows.
			\begin{enumerate}[leftmargin=*, itemsep=0pt, topsep=0pt]
				\item
				The prover computes the folded polynomial $g_w({\bf x})=g_{w_l}({\bf x})+g_{w_r}({\bf x})$.	
				
				\item
				For every $\vec{\cal V}_{w_l}[a]$, the prover provides $D_{w_r}$ and $y_{w_r, a}$ and for every $\vec{\cal V}_{w_r}[a]$, the prover provides $D_{w_l}$ and $y_{w_l, a}$. 
				
				\item
				Next, every $\vec{\cal V}_{w}[a]$ computes $y_{w, a}=y_{w_l, a}+y_{w_r, a}$ and $D_w=D_{w_l}D_{w_r}$. 
			\end{enumerate}
		\end{boxedminipage}
	\end{figure*}
	At the root, the final state is 
	$
	{\cal S}^* =
	\big({\cal P}: \{g^*\},\;
	\vec{\cal V}^*: \{(y^*_a, D^*)\}_{a \in [0,\mu)}\big),
	$
	where the prover holds the final folded polynomial $g^*$, and each verifier $\mathcal{V}_{j\mu+a}$ (for $j\in [0,\mu), a\in [0,\mu)$) holds the claimed evaluation $y^*_a = g^*({\sf Bin}(a))$ together with the commitment $D^*$.
	Finally, the prover runs the ${\sf PC.HyperEval}$ algorithm for $g^*$ to generate proofs $\{\pi^*_a\}_{a\in[0,\mu)}$ proving that each $y^*_a = g^*({\sf Bin}(a))$, and sends $\pi^*_a$ to verifiers $\{\mathcal{V}_{j\mu+a}\}_{j\in [0, \mu)}$.
	Verifiers $\{\mathcal{V}_{j\mu+a}\}_{j\in [0, \mu)}$ validates by running ${\sf PC.Verify}(D^*, {\sf Bin}(a), y_a^*, \pi^*_a)$.
\end{itemize}

\noindent
{\bf Making our protocol non-interactive.}
%To establish common randomness in the non-interactive setting, we use the Fiat-Shamir variant described in \cite{ZXH+22}, specifically designed for the multi-verifier setting and proven secure in the random oracle model. 
%Next, we explain the Fiat-Shamir variant at a high level.
%The prover builds the Merkle tree by hashing the invidiual statement at the leafs (together with the common statement) and uses its root as the common random challenge. Then, the proof provided to each verifier includes a membership proof in the Merkle tree to show that its data is accounted for. We refer the reader to \cite{ZXH+22} for additional details.
%For {\it Step 2} in our VC scheme, the prover builds the Merkle tree by hash the coefficients of all to-be-folded polynomials and their commitments.
Up to this point, our protocol has been presented in the interactive setting: in {\em Step 1}, the prover sends a message to the verifiers; in {\em Step 2}, the prover receives their challenges and responds accordingly. To make it non-interactive, we adopt the Fiat-Shamir variant of \cite{ZXH+22}, which is designed for the multi-verifier setting and proven secure in ROM.
Our protocol becomes non-interactive as follows.
For {\it Step 2}, the prover builds a Merkle tree by hashing the coefficients of all polynomials to be folded and their corresponding commitments. Each verifier then receives a proof that includes a Merkle tree membership path, showing that its data is correctly included.
%To establish common randomness in the non-interactive setting, we adopt the Fiat–Shamir variant of \cite{ZXH+22}, designed for the multi-verifier setting and proven secure in ROM. Here, the prover builds a Merkle tree by hashing each individual statement (along with the common statement) at the leaves, and then uses the Merkle root as the common random challenge; each verifier receives a proof that includes a Merkle tree membership path, showing that its data is correctly included. We refer the reader to~\cite{ZXH+22} for further technical details. Based on this, we make our protocol non-interactive as follows. For {\it Step 2} of {\bf Protocol 1}, the prover builds a Merkle tree by hashing the coefficients of all polynomials to be folded and their corresponding commitments.

The correctness and security of FlexProofs are proven in Appendix \ref{app:vc}. 
A zero-knowledge variant of {\em OpenAll} can also be obtained using standard techniques \cite{CFS17,L21} (see Appendix \ref{app:def} for the formal definition). In this case, the PC scheme is replaced with its zero-knowledge version, whereas the FC scheme does not require zero-knowledge since it only operates on already hidden commitments.

\subsection{Efficiency} 
\vspace{2mm}
\noindent
{\bf Efficiency improvements in the first step of {\bf Protocol 1}.}
Note that the working process of the first step in the protocol implies the \emph{first} subvector commitment scheme over group elements. Below, we describe some efficiency improvements for this scheme.
In the FC model, the vector length is denoted by $n$; in our setting this length is $\mu$, so we set $\ell=\log \mu$.
As a starting point, we consider the case where ${\bf b}$ in Eq.~\eqref{eq:folding} is the $i$-th {\em unit vector} ${\bf u}_i$.
Note that ${\sf Bin}(i)=(i_{\ell-1}, \ldots, i_0)$.
In the first folding round:
$$
({\bf u}_i)_1=({\bf u}_i)_L+1/x_1 \cdot({\bf u}_i)_R.
$$
If $i_{\ell-1} = 0$, the only non-zero entry $1$ appears in $({\bf u}_i)_L$, resulting in a single non-zero entry of $1$ in $({\bf u}_i)_1$.  
Conversely, if $i_{\ell-1} = 1$, the only non-zero $1$ entry appears in $({\bf u}_i)_R$, which leads to a single non-zero entry of $1/x_1$ in $({\bf u}_i)_1$.  
In summary, the only non-zero entry in $({\bf u}_i)_1$ is $(1/x_1)^{i_{\ell-1}}$, located at the index $(i_{\ell-2}, \ldots, i_0)$.
In general, at the $j$-th round, the only non-zero entry in $({\bf u}_i)_{j}$ is
$
\prod_{k=1}^{j} \left( {1}/{x_k} \right)^{i_{\ell-k}},
$
with index $(i_{\ell-j-1}, \ldots, i_0)$. 
Hence, after $\ell$ rounds, we have:
\begin{equation}\label{eq:b_l}
({\bf u}_i)_\ell = \prod_{k=1}^{\ell} \left({1}/{x_k} \right)^{i_{\ell-k}}.
\end{equation}
Let $I \subseteq [0, \mu)$ be a set of indices. For batch opening over $I$, ${\bf b}=\sum_{i\in I}r_i{\bf u}_i$.
Let $\hat{{\bf u}}_i=r_i{\bf u}_i(i\in I)$, thus
${\bf b}=\sum_{i\in I}\hat{{\bf u}}_i.$
According to the folding of ${\bf b}$ in Eq. \eqref{eq:folding}, at the first round:
${\bf b}_1={\bf b}_L+1/x_1{\bf b}_R=(\sum_{i\in I}\hat{{\bf u}}_i)_L+1/x_1(\sum_{i\in I}\hat{{\bf u}}_i)_R=\sum_{i\in I}(\hat{{\bf u}}_i)_L+1/x_1(\hat{{\bf u}}_i)_R=\sum_{i\in I}(\hat{{\bf u}}_i)_1;$
at the $j$-th round:
$
{\bf b}_j=\sum_{i\in I}((\hat{{\bf u}}_i)_{j-1})_L+1/x_j((\hat{{\bf u}}_i)_{j-1})_R
=\sum_{i\in I}(\hat{{\bf u}}_i)_j.
$
After $\ell$ rounds and by Eq.~\eqref{eq:b_l}:
$${\bf b}_\ell=\sum_{i\in I}(\hat{{\bf u}}_i)_\ell \text{ and }
(\hat{{\bf u}}_i)_\ell = r_i\prod_{k=1}^{\ell}(1/x_k)^{i_{\ell-k}}.$$
Therefore, ${\sf FC.BOpen}$ can run faster, and ${\sf FC.BVerify}$ can run in  $\mathcal{O}(|I| \log \mu)$ time.

\vspace{2mm}
\noindent
{\bf Complexity analysis of FlexProofs.}
The underlying homomorphic PC scheme for multi-linear polynomials achieves commitment and evaluation complexity of $\mathcal{O}(N)$, proof size of $\mathcal{O}(\log N)$, verification time of $\mathcal{O}(\log N)$, and $\sf PC.HyperEval$ complexity of $\mathcal{O}(N\log N)$.
Next, we analyze the complexity of FlexProofs.

\begin{itemize}
	\item
	{\em Commitment complexity.}
	Algorithm $\mathsf{VC.Commit}$ makes $\mu$ calls to $\mathsf{PC.Commit}$, each on a polynomial of size $\mu$, and then makes one call to $\mathsf{FC.Commit}$ on a vector of size $\mu$.
	Due to the linear commit complexity of both the PC scheme and our FC scheme, the overall commitment complexity is $\mathcal{O}(N)$.
	
	\item
	{\em Prover complexity.}
	In {\color{black}{\it Step 1}}, the prover creates proofs for $\lceil \mu / b \rceil$ batches, each containing $b$ elements. According to Table~\ref{tb:fcEff}, this step takes $\mathcal{O}(N)$ field operations and $\mathcal{O}(N / b)$ cryptographic operations.
	In {\color{black}{\it Step 2}}, the prover has to (1) randomize polynomials and their commitments, (2) fold polynomials in a tree-like structure, (3) call $\sf PC.HyperEval$.
	Due to the underlying PC scheme achieves $\sf PC.HyperEval$ complexity of $\mathcal{O}(N\log N)$, both (1) and (2) require $\mathcal{O}(N)$ field operations and $\mathcal{O}(\mu)$ cryptographic operations, (3) has a complexity of $\mathcal{O}(\sqrt{N}\log N)$.
	{\color{black}Overall, the prover runs in time $\mathcal{O}(N)$, with ${\cal O}(N)$ field operations and ${\cal O}(N/b + \sqrt{N}\log N)$ cryptographic operations.}
	
	\item
	{\em Proof size.}
	In {\it Step 1}, each ${\cal V}_i$ receives ${\bf C}_{\lfloor i/\mu b \rfloor}$ of size ${\cal O}(b)$ and the batch proof $\pi_{{\bf C}_{\lfloor i/\mu b \rfloor}}$ of size $\mathcal{O}(\log N)$.
	In {\it Step 2}, it receives $\log N$ commitments, $\log N$ claimed evaluations and an opening proof of the PC scheme. 
	Since the PC scheme has a proof size of ${\cal O}(\log N)$, the proof size would be ${\cal O}(\log N+b)$ for the interactive version and ${\cal O}(\log N+b)$ for the non-interactive version.
	
	\item
	{\em Verification time.}
	In {\it Step 1}, each ${\cal V}_i$ runs the $\sf FC.BVerify$ algorithm to verify the correctness of ${\bf C}_{\lfloor i/\mu b \rfloor}$, which takes ${\cal O}(b\log \mu)$ time.
	In {\it Step 2}, it needs to (1) compute the final evaluation and commitment, and (2) verify the final evaluation. 
	Since the underlying PC scheme has a verification time of ${\cal O}(\log N)$, the complexities of (1) and (2) are $\mathcal{O}(\log \mu)$ and ${\cal O}(\log N)$, respectively.
	Thus, the verification time would be $\mathcal{O}(b\log N)$ for both the interactive and non-interactive versions.
\end{itemize}
In conclusion, FlexProofs is a VC scheme with $\mathcal{O}(N)$ $\sf Commit$ and ${\sf OpenAll}$ complexity, $\mathcal{O}(\log N+b)$ proof size and $\mathcal{O}(b\log N)$ verification time.

\subsection{Compatible with a family of zkSNARKs}
Paper \cite{PPP25} proposes a construction that combines a VC scheme with a zkSNARK in order to prove the correctness of computations over data originating from multiple sources. To adapt FlexProofs to this construction, we first extend our {\em OpenAll} for sub-arrays and then show how FlexProofs is directly compatible with a family of zkSNARKs.

\vspace{2mm}
\noindent
{\bf Extending {\em OpenAll} for sub-arrays.}
The above {\em OpenAll} protocol proves the correctness of all individual elements of ${\bf m}$. We now extend this to proving the correctness of all {\em consecutive sub-arrays}. More formally, for each verifier ${\cal V}_j$ with $j \in [0,M)$, we aim to prove that its array ${\bf m}_j \in \mathbb{F}^{N/M}$ is indeed the $j$-th sub-array of the committed vector ${\bf m} \in \mathbb{F}^N$. This can be achieved by showing that the commitment $C_j$ of ${\bf m}_j$ is the $j$-th element of the committed vector ${\bf C}$, as established in {\em Step 1} of {\em OpenAll}. 
To generate all opening proofs in this setting, all parties follow the ${\sf VC.Commit}$ algorithm and {\em Step 1}, with the following modifications. In the ${\sf VC.Commit}$ algorithm, partition ${\bf m}$ into $M$ segments of size $N/M$ such that ${\bf m}=({\bf m}_0,\ldots,{\bf m}_{M-1})$; for each ${\bf m}_j$, compute its multi-linear extension commitment $C_j$; then compute the commitment of the vector ${\bf C}=(C_0,\ldots,C_{M-1})$. Upon receiving $C_j$, verifier ${\cal V}_j$ checks whether the polynomial committed in $C_j$ corresponds to the multi-linear extension of ${\bf m}_j$.

\vspace{2mm}
\noindent
{\bf Compatibility with a family of zkSNARKs.}
According to~\cite{PPP25}, if a VC scheme encodes the committed vector ${\bf m}$ as a multi-linear polynomial and supports generation of an evaluation proof for the polynomial at a random point, it is directly compatible with a family of zkSNARKs that encode inputs as multi-linear polynomials \cite{XZZ+19,S20,CBB+23}. 
Since FlexProofs is built on PC and FC schemes, it can be viewed as encoding the committed vector ${\bf m}$ into a multi-linear polynomial and supports generating an evaluation proof for the committed polynomial at a random point. 
We next show how FlexProofs produces such a proof.
By Eq. \eqref{eq:multi-extension}, the multi-linear extension of a vector ${\bf m}\in \mathbb{F}_p^N$ is 
\begin{equation}\label{eq:f_m}
f_{\mathbf{m}}(\mathbf{x}) =
\sum_{i\in [0, N)} {\bf m}[i]  \prod_{k\in [0, \log N)} \left( i_k x_k  + (1 - i_k)(1 - x_k) \right), 
\end{equation}
where $\mathbf{x} = (x_{\log N - 1}, \ldots, x_0)$ and ${\sf Bin}(i) = (i_{\log N - 1}, \ldots, i_0)$;
and for $j\in [0, \mu)$, the multi-linear extension of ${\bf m}_j\in \mathbb{F}_p^{\mu}$ is 
\begin{equation}\label{eq:f_j}
	f_j({\bf x}_R)=\sum_{a\in [0, \mu)} {\bf m}_j[a] \prod_{k\in [0, \log \mu)} \left( a_k x_k  + (1 - a_k)(1 - x_k) \right),
\end{equation}
where ${\sf Bin}(a)=(a_{\log \mu-1}, \ldots, a_0)$.
From Eq. \eqref{eq:f_m} and \eqref{eq:f_j}, we have 
\begin{equation}\label{eq:fANDfj}
f_{\bf m}({\bf x})=\sum\nolimits_{j\in [0, \mu)}f_j({\bf x}_R)\mathcal{T}_{j, \log \mu}({\bf x}_L),
\end{equation}
where $\mathcal{T}_{j, \log \mu}({\bf x}_L)=\prod_{a\in [0, \log \mu)}(j_ax_{a+\log \mu}+(1-j_a)(1-x_{a+\log \mu}))$
and ${\sf Bin}(j)=(j_{\log \mu-1}, \ldots, j_0)$.
To prove the evaluation of $f_{\bf m}({\bf x})$ at ${\bf r}=(r_{\log N-1}, \ldots, r_0)\in \mathbb{F}_p^{\log N}$, it suffices to prove the evaluation of $f_{\bf m}({\bf r}_L, {\bf x}_R)$ at ${\bf r}_R$. 
Let us first define 
$F({\bf x}_R)=f_{\bf m}({\bf r}_L, {\bf x}_R);$
using Eq. \eqref{eq:fANDfj},  we then obtain the expression: 
$$F({\bf x}_R)=\sum\nolimits_{j\in [0, \mu)}f_j({\bf x}_R)\mathcal{T}_{j, \log \mu}({\bf r}_L).$$
Thus, the original problem reduces to proving the evaluation of $F({\bf x}_R)$ at ${\bf r}_R$, which can be accomplished through the following steps:
\begin{enumerate}
	\item Compute the commitment to $F(\mathbf{x}_R)$ as $C_F=\prod_{j\in [0, \mu)} C_j^{\mathcal{T}_{j, \log \mu}({\bf r}_L)}$.
	\item Prove the correctness of $C_F$ using our FC scheme as $C_F$ can be expressed as $\langle {\bf C}, (\mathcal{T}_{0, \log \mu}({\bf r}_L), \ldots, \mathcal{T}_{\mu-1, \log \mu}({\bf r}_L))\rangle$, where ${\bf C}$ is already committed in ${C}$.
	\item Generate a proof for the evaluation of $F$ at ${\bf r}_R$ with the PC scheme.
\end{enumerate}
Therefore, to prove the evaluation of $f({\bf x})$ at a random point ${\bf r}$, the prover sends $C_F$ along with proofs generated by the FC and PC schemes to the verifier, who then checks both proofs.
The above process can be adapted into a zero-knowledge version via standard techniques \cite{CFS17,L21}. Here, the underlying PC is replaced with its zero-knowledge variant, while the FC component does not require zero-knowledge, as it only operates on hidden commitments.

\section{Performance Evaluation}
In this section, we implement the non-interactive variants of FlexProofs \footnote{https://github.com/FlexProofs} in the single-thread and evaluate their performance. 
For field and group operations, we use the mcl library and choose BN\_SNARK1 elliptic curve to achieve approximately 100 bits of security.
%A $\mathbb{G}_1$, $\mathbb{G}_2$ and $\mathbb{G}_T$ element takes 32, 64 and 384 bytes, respectively.
Our codes are run on an Intel(R) Xeon(R) E-2286G CPU @ 4.0GHz with 6 cores and 64GB memory. 
Unless otherwise stated, we run each experiment 3 times and report their average.
%We run each experiment 3 times and report their average.

\subsection{Evaluation of our FC scheme}
We evaluate our FC scheme for $n\in \{2^{8}, 2^{9}, 2^{10}, 2^{11}, 2^{12}\}$, and show the results in Table \ref{tb:fc}.

\vspace{-8mm}
\renewcommand{\arraystretch}{1.1}
\begin{table}[htbp]
	\centering
	\caption{Single-threaded running time of our FC scheme.}
	\begin{tabular}{c| C{1cm} C{1cm} C{1cm} C{1cm} C{1cm}}
		\hline
		$n$& {$2^{8}$} & {$2^{9}$} & {$2^{10}$} & {$2^{11}$} & {$2^{12}$} \\
		\hline
		|$\pi_y$| (KiB)&6.66&7.47&8.28&9.09&9.91\\
		%\hline
		%Proof size (KiB)&&&&&\\
		\hline
		${\sf FC.Commit}$ (s)&0.03&0.07&0.14&0.27&0.54\\
		\hline
		${\sf FC.BOpen}(t=1)$ (s)&0.12&0.23&0.45&0.88&1.70\\
		\hline
		${\sf FC. BVerify}(t=1)$ (s)&0.004&0.005&0.005&0.006&0.007\\
		\hline
		${\sf FC.BOpen}(t=32)$ (s)&0.13&0.23&0.46&0.89&1.73\\
		\hline
		${\sf FC. BVerify}(t=32)$ (s)&0.007&0.01&0.02&0.03&0.05\\
		\hline
	\end{tabular}
	\label{tb:fc}
\end{table}
\vspace{-8mm}

\vspace{2mm}
\noindent
{\bf Proof size.}
The proof size grows with $\log n$, and is about 9.91\,KiB for $n = 2^{12}$.

\vspace{2mm}
\noindent
{\bf Commit.}
As shown in Table \ref{tb:fc}, the time costs of ${\sf FC.Commit}$ are roughly proportional to $n$.
For $n=2^{12}$, committing needs $0.54$s.

\vspace{2mm}
\noindent
{\bf Batch opening.}
For $n = 2^{12}$, the running time of ${\sf FC.BOpen}(t = 32)$ is approximately 1.73\,s, which is only 3\% of the total time required to run ${\sf FC.BOpen}(t = 1)$ 32 times. 
{\color{black}
Figure \ref{fig:BatchOpening} compares the running time of a single ${\sf FC.BOpen}$ call with batch size $t$ (denoted ${\sf FC.BOpen}$-$t$) against $t$ sequential calls with batch size 1 (denoted $t\times {\sf FC.BOpen}$-$1$), for $n = 2^{12}$ and $t \in \{16, 32, 64, 128, 256, 512\}$.}
And it shows that the efficiency gain of batching increases with the batch size $t$.

\vspace{2mm}
\noindent
{\bf Batch verification.}
For $n = 2^{12}$, the running time of ${\sf FC.BVerify}(t = 32)$ is approximately 0.05\,s, which is only 22\% of the total time required to run ${\sf FC.BVerify}(t = 1)$ 32 times. This demonstrates the high efficiency of our proposed batching technique.
{\color{black}
Figure \ref{fig:BatchVerification} compares the running time of a single ${\sf FC.BVerify}$ call with batch size $t$ (denoted ${\sf FC.BVerify}$-$t$) against $t$ sequential calls with batch size 1 (denoted $t\times {\sf FC.BVerify}$-$1$), for $n = 2^{12}$ and $t \in \{16, 32, 64, 128, 256, 512\}$.}

\begin{figure}[htbp]
	\centering
	\begin{minipage}[t]{0.48\textwidth}
		\includegraphics[width=\linewidth]{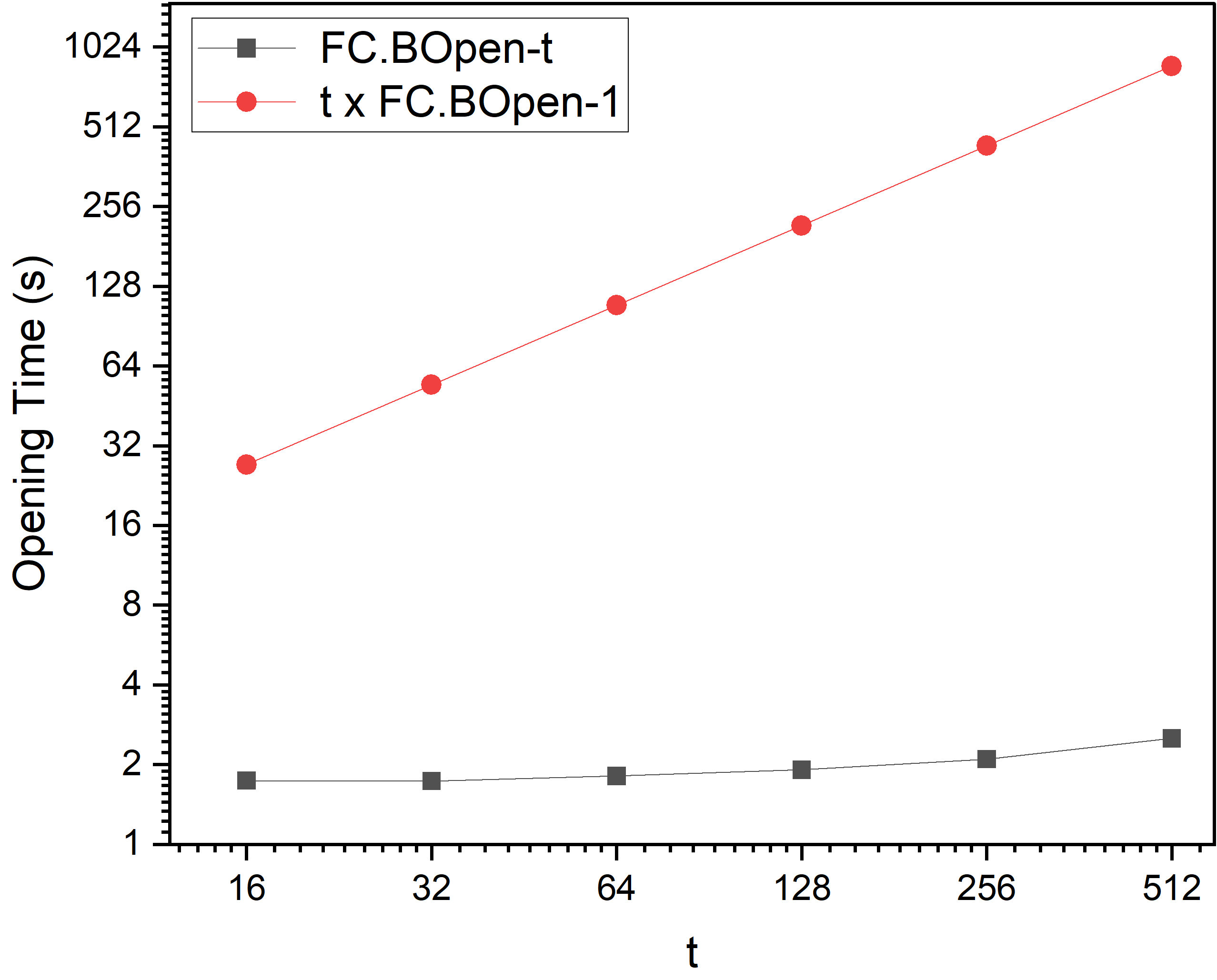}
		\caption{\color{black} Opening time. Both axes use logarithmic scales.}%(${\sf FC.BOpen}$-$t$ denotes a single call to ${\sf FC.BOpen}$ with batch size $t$; $t\times {\sf FC.BOpen}$-$1$ denotes $t$ sequential calls, each with batch size 1.)
		\label{fig:BatchOpening}
	\end{minipage}
	\hfill 
	\begin{minipage}[t]{0.48\textwidth}
		\includegraphics[width=\linewidth]{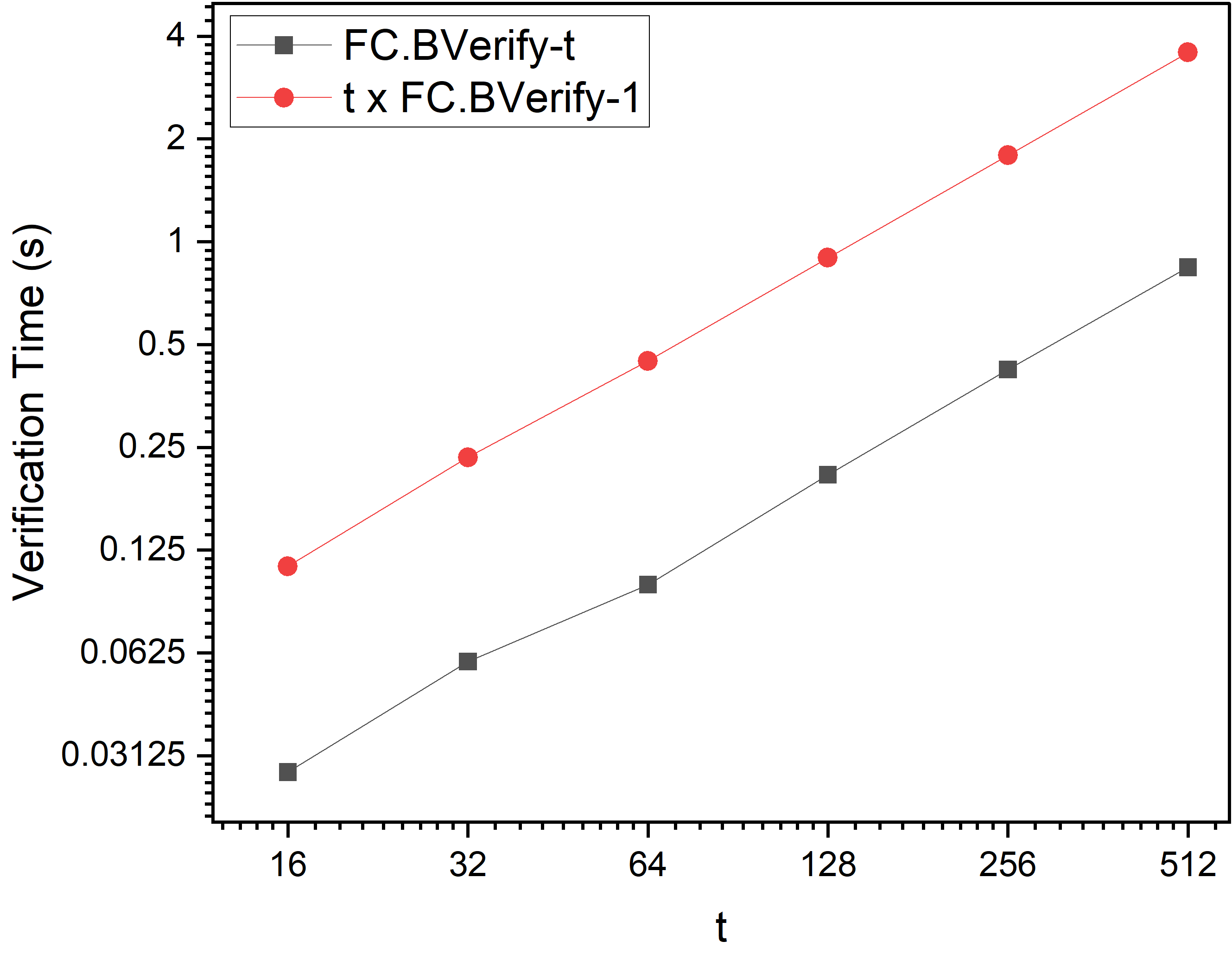}
		\caption{\color{black} Verification time. Both axes use logarithmic scales.}%(${\sf FC.BVerify}$-$t$ denotes a single call to ${\sf FC.BVerify}$ with batch size $t$; $t\times {\sf FC.BVerify}$-$1$ denotes $t$ sequential calls, each with batch size 1.) 
		\label{fig:BatchVerification}
	\end{minipage}
\end{figure}

\vspace{-10mm}
\subsection{Evaluation of FlexProofs and comparison with HydraProofs}
We microbenchmark FlexProofs with batch sizes $b \in \{2\log N, \log^{2} N\}$. 
We also compared it with HydraProofs \cite{PPP25}, the state-of-the-art vector commitment scheme. 
HydraProofs is the only known VC construction that achieves ${\cal O}(N)$ time to generate all opening proofs for a vector of size $N$, while also being directly compatible with a family of zkSNARKs. 
For fairness, we implement both schemes in C++ using
the same elliptic curve and techniques for $N\in \{2^{16}, 2^{18}, 2^{20}, 2^{22}, 2^{24}\}$. The experimental results are shown in Table \ref{tb:vc}.

\begin{table}[htbp]
	\small
	\centering
	\caption{The comparison between FlexProofs (FP) and HydraProofs.}
	\begin{tabular}{c| c | C{1cm}C{1cm}C{1cm}C{1cm}C{1cm}}
		\hline
		\multicolumn{2}{c|}{$N$} & {$2^{16}$} & {$2^{18}$} & {$2^{20}$} & {$2^{22}$} & {$2^{24}$}\\
		\hline
		Commit (s)
		&FP ($b=2\log N$) &0.83&3.02&8.81&31.87&116.61\\
		&FP ($b=\log^2 N$) &0.82&3.01&8.70&31.53&116.41\\
		&HydraProofs&0.82&3.01&8.66&31.71&116.64 \\
		\hline
		{Compute all}  
		&FP ($b=2\log N$) & 1.03 &3.58&12.17&43.64&159.30\\
		{proofs} (s)
		&FP ($b=\log^2 N$) & {\bf 0.23} &{\bf 0.82}&{\bf 2.61}&{\bf 9.24}&{\bf 32.71}\\
		&HydraProofs& 1.42 &4.82&17.15&58.66&210.07\\
		\hline
		{Verify a}
		&FP ($b=2\log N$) &{\bf 0.006}&{\bf 0.007}&{\bf 0.008}&{\bf 0.009}&{\bf 0.011}\\
		{proof} (s)
		&FP ($b=\log^2 N$) & 0.01 &0.012&0.013&0.015&0.017\\
		&HydraProofs &0.01&0.011&0.012&0.013&0.014\\
		\hline
		{Proof size}
		&FP ($b=2\log N$)& 8.91 & 10 & 11.09 & 12.19 & 13.28\\
		(KiB)&FP ($b=\log^2 N$) & 15.91 &19&22.34&25.94&29.78\\
		&HydraProofs& {\bf 5.53} & {\bf 6.63} & {\bf 7.81} & {\bf 9.09} & {\bf 10.47} \\
		\hline
	\end{tabular}
	\label{tb:vc}
\end{table}

\vspace{2mm}
\noindent
{\bf Commit.}
In FlexProofs, the commitment time grows roughly linearly with $N$, taking about $0.82$ seconds when $N=2^{16}$. 
HydraProofs shows a similar linear growth in commitment time.

\vspace{2mm}
\noindent
{\bf Computing all proofs.}
{\color{black}In FlexProofs, computing all proofs requires ${\cal O}(N)$ field operations and ${\cal O}(N/b + \sqrt{N}\log N)$ cryptographic operations, so the proving time decreases as $b$ grows.
Experiments confirm this: when $N=2^{16}$, computing all proofs takes $1.03$ seconds for $b=2\log N$, 
but only $0.23$ seconds for $b=\log^2 N$. 
In HydraProofs, computing all proofs requires ${\cal O}(N)$ field operations and ${\cal O}(N)$ cryptographic operations.
Thus, FlexProofs requires fewer cryptographic operations as $b$ increases. Consequently, when $N=2^{16}$, FlexProofs is about $1.3\times$ faster than HydraProofs with $b=2\log N$, and about $6\times$ faster with $b=\log^2 N$.}

%Although computing all proofs in FlexProofs has overall complexity ${\cal O}(N)$, a larger batch size $b$ reduces the number of cryptographic operations and thus lowers the proving time. 
%In HydraProofs, the proving time also has complexity ${\cal O}(N)$. However, while both schemes share the same asymptotic complexity, FlexProofs requires fewer cryptographic operations as $b$ increases. 

\vspace{2mm}
\noindent
{\bf Verifying a proof.}
When $b=2\log N$, the verification complexity in FlexProofs is ${\cal O}(\log^2 N)$, 
whereas for $b=\log^2 N$ it increases to ${\cal O}(\log^3 N)$. 
Experimental results confirm this: when $N=2^{16}$, verifying a proof takes $0.006$ seconds for $b=2\log N$, and $0.01$ seconds for $b=\log^2 N$. 
In HydraProofs, the verification complexity is also ${\cal O}(\log^2 N)$. 
And the experimental results show that when $b=2\log N$, FlexProofs verifies proofs slightly faster than HydraProofs.

\vspace{2mm}
\noindent
{\bf Proof size.}
When $b=2\log N$, the proof size in FlexProofs is ${\cal O}(\log N)$, 
while for $b=\log^2 N$ it grows to ${\cal O}(\log^2 N)$. 
In HydraProofs, the proof size has complexity ${\cal O}(\log^2 N)$. 
For the values of $N$ shown in Table~\ref{tb:vc}, HydraProofs yields smaller proofs; 
for instance, when $N=2^{16}$ it requires only $5.53$ KiB compared to $8.91$ KiB in our scheme. 
This gap arises because asymptotic bounds capture only the growth trend as $N$ increases, 
while for moderate $N$ the hidden constants dominate the concrete size. 
As $N$ becomes larger, the slower asymptotic growth of FlexProofs is expected to eventually result in smaller proofs.

%\vspace{2mm}
%\noindent
%{\bf Proof size.}
%The proof size would be $\mathcal{O}(\log N)$.

\section{Conclusions}
We first propose an FC scheme for multi-exponentiations with batch opening. Based on this, we construct FlexProofs, a VC scheme that can generate all proofs for a vector of size $N$ in optimal time ${\cal O}(N)$ (incurred by ${\cal O}(N)$ field operations and ${\cal O}(N/b + \sqrt{N}\log N)$ cryptographic operations) and is directly compatible with zkSNARKs, where batch size $b$ ranges from $1$ to $\sqrt{N}$. 
Finally, FlexProofs, combined with suitable zkSNARKs, enables applications such as VSS and VRA.

\section*{Acknowledgments}
The authors would like to thank the anonymous shepherd and reviewers for their insightful comments. The corresponding author Liang Feng Zhang's research is partially supported by the National Natural Science Foundation of China (Grant No. 62372299).
\bibliographystyle{splncs04}
\bibliography{references}

\appendix
\section{Definitions}\label{app:def}
\begin{definition}
[\bf Zero-Knowledge of Polynomial Commitment]
A polynomial commitment (PC) scheme is zero-knowledge if for any $\lambda$, $n$, $d$, adversary $\mathcal{A}$, and simulator $\mathcal{S}$, we have:
$$
\Pr\left(\text{Real}_{\mathcal{A}, f}(1^\lambda) = 1\right) \approx \Pr\left(\text{Ideal}_{\mathcal{A}, \mathcal{S}}(1^\lambda) = 1\right).
$$
\fbox{
\begin{minipage}[t]{0.48\textwidth}
	\em
	\noindent\textbf{Real$_{\mathcal{A}, f}(1^\lambda)$}:
	\begin{enumerate}
		\item ${\sf pp}\leftarrow {\sf PC.Setup}(1^\lambda, 1^n, 1^d)$
		\item $C_f \leftarrow {\sf PC.Commit}(f, r_f)$
		\item $m \leftarrow \mathcal{A}(1^\lambda, {\sf pp}, C_f)$
		\item For each step $j \in \{2, 3, \ldots, m\}$:
		\begin{enumerate}
			\item $\mathbf{x}_j \leftarrow \mathcal{A}(1^\lambda, C_f, y_1, \ldots, y_{j-1},\\ \pi_1, \ldots, \pi_{j-1}, {\sf pp})$
			\item $y_j, \pi_j \leftarrow {\sf PC.Eval}(f, x_j)$
		\end{enumerate}
		\item Output $b \leftarrow \mathcal{A}(1^\lambda, C_f, y_1, \ldots, \\ y_m, \pi_1, \ldots, \pi_m, {\sf pp})$
	\end{enumerate}
\end{minipage}
\hfill
\begin{minipage}[t]{0.48\textwidth}
	\noindent\textbf{Ideal$_{\mathcal{A}, \mathcal{S}}(1^\lambda)$}:
	\begin{enumerate}
		\item $(C_f, {\sf pp}, t) \leftarrow \mathcal{S}(1^\lambda, 1^n)$
		\item $m \leftarrow \mathcal{A}(1^\lambda, {\sf pp}, C_f)$
		\item For each step $j \in \{2, 3, \ldots, m\}$:
		\begin{enumerate}
			\item ${\bf x}_j \leftarrow \mathcal{A}(1^\lambda, C_f, y_1, \ldots, y_{j-1}, \\ \pi_1, \ldots, \pi_{j-1}, {\sf pp})$
			\item $y_j, \pi_j \leftarrow \mathcal{S}({\sf pp}, f, {\bf x}_j)$
		\end{enumerate}
		\item Output $b \leftarrow \mathcal{A}(1^\lambda, C_f, y_1, \ldots, \\ y_m, \pi_1, \ldots, \pi_m, {\sf pp})$
	\end{enumerate}
\end{minipage}}
\end{definition}

\begin{definition}
	[Position Hiding of Vector Commitment]
	A vector commitment (VC) scheme is position hiding if for any PPT adversary \(\mathcal{A}\) and simulator \(\mathcal{S}\), the following holds:
	\[
	\Pr\left(\text{Real}_{\mathcal{A}, {\bf m}}(1^\lambda) = 1\right) \approx \Pr\left(\text{Ideal}_{\mathcal{A}, \mathcal{S}}(1^\lambda) = 1\right).
	\]
	\fbox{
		\begin{minipage}[t]{0.48\textwidth}
			\em
			\noindent\textbf{ Real$_{\mathcal{A}, {\bf m}}(1^\lambda)$}:
			\begin{enumerate}
				\item ${\sf pp} \leftarrow {\sf VC.Setup}(1^\lambda, 1^N)$
				\item $C \leftarrow {\sf VC.Commit}(\mathbf{m})$
				\item $i_1, \ldots, i_n, m_{i_1}, \ldots, m_{i_n} \leftarrow\\ \mathcal{A}(1^\lambda, {\sf pp},  C)$ (for $n < N$)
				\item $\{\pi_i\}_{i=1}^{N} \leftarrow {\sf VC.OpenAll}(\mathbf{m})$.\\ Send $ \{\pi_{i_1}, \ldots, \pi_{i_n}\}$ to $\mathcal{A}$
				\item Output $b \leftarrow \mathcal{A}(1^\lambda, C, m_{i_1}, \\ \ldots,  m_{i_n}, \pi_{i_1}, \ldots, \pi_{i_n}, {\sf pp})$
			\end{enumerate}
		\end{minipage}
		\hfill
		\begin{minipage}[t]{0.48\textwidth}
			\noindent\textbf{Ideal$_{\mathcal{A}, \mathcal{S}}(1^\lambda)$}:
			\begin{enumerate}
				\item $(C, {\sf pp}, \text{trap}) \leftarrow \mathcal{S}(1^\lambda, N)$
				\item $i_1, \ldots, i_n, m_{i_1}, \ldots, m_{i_n} \leftarrow \\ \mathcal{A}(1^\lambda, {\sf pp}, C)$ (for $n < N$)
				\item $\{\pi_{i_1}, \ldots, \pi_{i_n}\} \leftarrow  \mathcal{S}(i_1, \ldots, i_n, m_{i_1}, \\ \ldots, m_{i_n}, \text{trap}, {\sf pp})$
				\item Output $b \leftarrow \mathcal{A}(1^\lambda, C, m_{i_1}, \ldots, \\ m_{i_n}, \pi_{i_1}, \ldots, \pi_{i_n}, {\sf pp})$
			\end{enumerate}
	\end{minipage}}
\end{definition}

\section{Our FC Scheme}\label{app:B}
\begin{assumption}
	[$q$-Strong Diffie-Hellman assumption ($q$-SDH)\cite{BB08}]\label{asp:qSDH} 
	Given a security parameter $\lambda$, algorithm ${\sf BG}$ outputs ${\bf bg}=$ 
	$(p, \mathbb{G}_1, \mathbb{G}_2, \mathbb{G}_{\rm T}, e, g_1, g_2)$.
	The $q$-SDH assumption holds relative to ${\sf BG}$ if for any efficient algorithm $\mathcal{A}$,
		$$
		\Pr
		\begin{bmatrix}
		{\bf bg} \leftarrow {\sf BG}(1^\lambda),
		\gamma \overset{\$}{\leftarrow} \mathbb{F}_p,
		{\sf pp}=({\bf bg}, g_1^\gamma, \ldots, g_1^{\gamma^q}, g_2^\gamma):\\
		(a, g_1^{{1}/({\gamma+a})})\leftarrow \mathcal{A}(1^\lambda, {\sf pp})
		\end{bmatrix}
		\leq {\sf negl}(\lambda).
		$$
\end{assumption}

\begin{assumption}
	[$q$-Auxiliary Structured Double Pairing assumption 
	]{\bf ($q$-ASDBP)}\label{asp:qASDBP}
	Given a security parameter $\lambda$, ${\sf BG}$ output ${\bf bg}=(p, \mathbb{G}_1, \mathbb{G}_2, \mathbb{G}_{\rm T}, e, g_1, g_2)$.
	The $q$-ASDBP assumption holds relative to ${\sf BG}$ if for any efficient algorithm $\mathcal{A}$,
	$$
	\Pr
	\begin{bmatrix}
	{\bf bg} \leftarrow {\sf BG}(1^\lambda),~
	\beta \overset{\$}{\leftarrow} \mathbb{F}_p,
	(A_0, \ldots, A_{q-1}) \leftarrow  {\mathcal{A}}({\bf bg}, g_1^\beta, \{g_2^{\beta^{2i}}\}_{i=1}^{q-1}):\\
	((A_0, \ldots, A_{q-1}) \neq {\bf 1}_{\mathbb{G}_1}) \wedge ({1}_{\mathbb{G}_{\rm T}}=\prod\nolimits_{i=0}^{q-1}e(A_i, g_2^{\beta^{2i}}))\\
	\end{bmatrix}
	\leq {\sf negl}(\lambda).
	$$
\end{assumption}

\begin{lemma}\label{lem:equalityB}
	Let ${\bf A}\in \mathbb{G}_1^n$. Let ${\bf b}^{(i)}\in \mathbb{F}_p^n$ and $y_i\in \mathbb{G}_1$ for $i\in [0, t)$.
	Assume each $r_i(i\in [0, t))$ is chosen uniformly at random from $\mathbb{F}_p$. Then with probability at least $1-1/p$, all Eq. \eqref{eq:distinctB} are satisfied iff Eq. \eqref{eq:batchB} is satisfied.
\end{lemma}

\begin{proof}
	Clearly, if Eq. \eqref{eq:distinctB} holds, then Eq. \eqref{eq:batchB} holds.
	For the other direction, assume the correct value of ${\bf A}^{{\bf b}^{(i)}}$ is $\hat{y}_i$ for $i\in [0, t)$, then there exists at least one $j$ such that $y_j\neq \hat{y}_j$ and $\prod_{i\in [0, t)} y_i^{r_i}=\prod_{i\in [0, t)} \hat{y}_i^{r_i}$. 	
	Define the  $a_i=\log_{g_1}y_i/\hat{y}_i$, then there will be at least one $j$ for which $a_j\neq 0$ and
	$
	\sum_{i\in [0, t)}a_i{r_i}=0 \mod p.
	$
	By the Schwartz-Zippel lemma \cite{S80,Z79}, this occurs with probability $\leq 1/p$.
	Therefore, if Eq. \eqref{eq:distinctB} is not fully satisfied, Eq. \eqref{eq:batchB} holds with probability at most $1/p$.\qed
\end{proof}

\vspace{2mm}
\noindent
{\bf Non-interactive argument of knowledge in the ROM.}
B{\"u}nz et al. \cite{BMM+21} define a non-interactive argument of knowledge in the random oracle model (ROM). 
In ROM, a hash function is replaced
by a random function which is sampled from the space of all random functions $\rho \leftarrow \mathcal{U}(1^\lambda)$. 
The non-interactive argument is an argument system where the prover sends a single message $\pi$, and the verifier using the proof accepts or rejects. Both the prover and verifier have access to a random oracle $\rho$. An argument of knowledge in the ROM has the property that for each convincing prover there exists an extractor which can rewind the prover and reinitialize the random oracle with new randomness.
\begin{definition}\label{df:non-interactive}
	{\bf (Non-interactive argument of knowledge in the ROM, from \cite{BMM+21}).}
	A non-interactive argument is an argument of knowledge for a relation $\mathcal{R}$ with knowledge error $\kappa(\lambda)$ if for every adversary $\widetilde{\mathcal{P}}$ there exists an extractor $\mathcal{E}$ such that
	$$
	\Pr
	\begin{bmatrix}
	\rho \leftarrow \mathcal{U}(1^\lambda),~
	{\sf crs} \leftarrow {\sf Setup}(1^\lambda),
	(\mathbbm{x}, \pi) \leftarrow  \widetilde{\mathcal{P}}^\rho({\sf crs}),\\
	\mathbbm{w} \leftarrow  {\mathcal{E}}^{\widetilde{\mathcal{P}}, \rho}({\sf crs}, \mathbbm{x}, \pi):
	(\mathcal{V}^\rho({\sf crs}, \mathbbm{x}, \pi)=1) \wedge ((\mathbbm{x}, \mathbbm{w})\notin {\cal R})
	\end{bmatrix}
	\leq \kappa(\lambda).
	$$
	The $\mathcal{E}$ can rewind the prover and reinitialize (but not program) the random oracle.
\end{definition}

\subsection{Function Binding}\label{sec:FCSnd}
Before proving the function binding of our FC scheme, we first show that our FC scheme is a non-interactive argument of knowledge in the ROM for the relation 
$\mathcal{R}_B=
\{({C}\in \mathbb{G}_{\rm T}, \{{\bf b}^{(i)}\}_{i\in [0, t)}\subset \mathbb{F}_p^n, \{y_i\}_{i\in [0, t)}\subset \mathbb{G}_1, {\bf v}\in \mathbb{G}_2^n; {\bf A}\in \mathbb{G}_1^n):({C}={\bf A}*{\bf v})\wedge (\forall i\in [0, t), y_i=\langle {\bf A}, {\bf b}^{(i)}\rangle)\}.
$
That is, we show that for any PPT adversary $\mathcal{A}_B$, we can construct a PPT extractor $\mathcal{E}_B$ that extracts the vector ${\bf A}$ such that
$$
\Pr
\begin{bmatrix}
\rho \leftarrow \mathcal{U}(1^\lambda),~
{\sf pp}(={\bf v}) \leftarrow {\sf FC.Setup}(1^\lambda, 1^n),\\
({C}, \{{\bf b}^{(i)}\}_{i\in [0, t)}, \{y_i\}_{i\in [0, t)}, \pi_y) \leftarrow  {\mathcal{A}}_B^\rho({\sf pp}),\\
{\bf A} \leftarrow  {\mathcal{E}}_B^{{\mathcal{A}_B}, \rho}({\sf pp}, {C}, \{{\bf b}^{(i)}\}_{i\in [0, t)}, \{y_i\}_{i\in [0, t)}, \pi_y):\\
({\sf FC.BVerify}^\rho({C}, \{{\bf b}^{(i)}\}_{i\in [0, t)}, \{y_i\}_{i\in [0, t)}, \pi_y)=1) \wedge \\
(({C}, \{{\bf b}^{(i)}\}_{i\in [0, t)}, \{y_i\}_{i\in [0, t)}, {\bf v}; {\bf A})\notin {\cal R}_B)
\end{bmatrix}
\leq {\sf negl}(\lambda).
$$

To prove this, we first prove the knowledge soundness of the $\sf FC.BOpen$ and $\sf FC.BVerify$ algorithms in the single-instance case (specifically, $t=1$) and without the random value $r_0$.

\vspace{2mm}
\noindent
{\bf Knowledge soundness in the single-instance case and without the random value.}
For simplicity, we denote the $\sf FC.BOpen$ and $\sf FC.BVerify$ algorithms in the single-instance case (specifically, $t=1$) and without the random value $r_0$ as $\sf FC.Open$ and $\sf FC.Verify$. 
Note that in $\sf FC.Open$ and $\sf FC.Verify$, ${\bf b}={\bf b}^{(0)}$ and $y=y_0$.
We will show that the scheme consisting of algorithms $\sf FC.Setup$, $\sf FC.Commit$, $\sf FC.Open$ and $\sf FC.Verify$ is a non-interactive argument of knowledge in the ROM for the relation $\mathcal{R}_S=\{({C}\in \mathbb{G}_{\rm T}, {\bf b}\in \mathbb{F}_p^n, y\in \mathbb{G}_1, {\bf v}\in \mathbb{G}_2^n; {\bf A}\in \mathbb{G}_1^n):{C}={\bf A}*{\bf v}\wedge y=\langle {\bf A}, {\bf b}\rangle\}.$
That is, we show that for any PPT adversary $\mathcal{A}_S$, we can construct a PPT extractor $\mathcal{E}_S$ that extracts the vector ${\bf A}$ such that
$$
\Pr
\begin{bmatrix}
\rho \leftarrow \mathcal{U}(1^\lambda),~
{\sf pp}(={\bf v}) \leftarrow {\sf FC.Setup}(1^\lambda, 1^n),\\
({C}, {\bf b}, y, \pi_y) \leftarrow  {\mathcal{A}}_S^\rho({\sf pp}),
~{\bf A} \leftarrow  {\mathcal{E}}_S^{{\mathcal{A}_S}, \rho}({\sf pp}, {C}, {\bf b}, y, \pi_y):\\
({\sf FC.Verify}^\rho({C}, {\bf b}, y, \pi_y)=1) \wedge 
(({C}, {\bf b}, y, {\bf v}; {\bf A})\notin {\cal R}_S)
\end{bmatrix}
\leq {\sf negl}(\lambda).
$$
The existence of $\mathcal{E}_S$ is guaranteed
by \cite{BMM+21}.
\begin{comment}
Since ${\sf FC.Verify}$ is correctly executed and outputs 1, we have 
${\bf C}_{j+1} = {\bf L}_{j+1}^{x_{j+1}}\circ {\bf C}_{j} \circ ({\bf R}_{j+1})^{1/x_{j+1}}$
for $j\in [0, l)$
and
${\bf C}_{\ell}=({\bf A}_{\ell}*{\bf v}_{\ell}, \langle {\bf A}_{\ell}, {{\bf b}_{\ell}}\rangle)$.
By Lemma \ref{le:statement}, the extractor $\mathcal{E}_S$ can extract the witness of ${\bf C}_{\ell-1},\ldots, {\bf C}_{0}$ in sequence with overwhelming probability.
Therefore, the extractor $\mathcal{E}_S$ extracts the witness ${\bf A}$ for ${\bf C}_0$ with overwhelming probability. 
\end{comment}

\vspace{2mm}
\noindent
{\bf Constructing ${\cal E}_B$.}
${\cal E}_B$ builds a adversary ${\cal A}_S$ against the extractability game of algorithms $\sf FC.Verify$. ${\cal A}_S$ is given $({\sf pp}, {C}, \{{\bf b}^{(i)}\}_{i\in [0, t)}, \{y_i\}_{i\in [0, t)}, \pi_y)$, computes ${\bf b}=\sum_{i\in [0, t)}r_i{\bf b}^{(i)}$ and $y=\prod_{i\in [0, t)} y_i^{r_i}$, and output $({C}, {\bf b}, y, \pi_y)$. Then ${\cal E}_B$ invokes ${\cal E}_S$ which outputs the vector ${\bf A}$ such that ${C}={\bf A}*{\bf v}$ and $y=\langle {\bf A}, {\bf b}\rangle$.
The equation $y=\langle {\bf A}, {\bf b}\rangle$ means that $\prod_{i\in [0, t)} y_i^{r_i}=\langle {\bf A}, \sum_{i\in [0, t)}r_i{\bf b}^{(i)}\rangle$. 
According to Lemma \ref{lem:equalityB}, with overwhelming probability, all $\{\langle {\bf A}, {\bf b}^{(i)}\rangle=y_i\}_{i\in [0, t)}$ are satisfied. 
Therefore, ${\cal E}_B$ can output vector ${\bf A}$ such that $({C}, \{{\bf b}^{(i)}\}_{i\in [0, t)}, \{y_i\}_{i\in [0, t)}, {\bf v}; {\bf A})\in {\cal R}_B$ with overwhelming probability.
The probability ${\cal E}_B$ fails is only negligible to the security parameter.

\vspace{2mm}
\noindent
{\bf Proving Function Binding.}
From the above, we prove the function binding of our FC scheme as follows. Assuming that an adversary $\cal A$ outputs $C$, $\{{\bf b}^{(i)}, y_{i}, \hat{y}_{i}\}_{i\in [0, t)}, \pi_{y}, \pi_{\hat{y}}$ such that (1) 
${\sf FC.BVerify}(C, \{{\bf b}^{(i)}\}_{i\in [0, t)}, \{{y}_{i}\}_{i\in [0, t)}, \pi_{{y}})=1$ and (2)
${\sf FC.BVerify}(C, \{{\bf b}^{(i)}\}_{i\in [0, t)}, \{\hat{y}_{i}\}_{i\in [0, t)}, \pi_{\hat{y}})=1$ with non-negligible probability.
By the knowledge soundness of our
FC scheme, equation (1) means that we can extract a vector ${\bf A}$ such that $C={\bf A}*{\bf v}$ and $\{y_i=\langle {\bf A}, {\bf b}^{(i)}\rangle)\}_{i\in [0, t)}$, equation (2) means that we can extract a vector ${\bf A}^\prime$ such that $C={\bf A}^\prime*{\bf v}$ and $\{\hat{y}_i=\langle {\bf A}^\prime, {\bf b}^{(i)}\rangle)\}_{i\in [0, t)}$.
Since there exists $j\in [0, t)$ such that $y_j\neq \hat{y}_j$, we have ${\bf A}\neq{\bf A}^\prime$.
Since ${\bf A}\neq {\bf A}'$ and ${\bf A}*{\bf v}={\bf A}'*{\bf v}$, we construct another adversary $\mathcal{B}$ which outputs $({\bf A}[i]/{\bf A}^\prime[i])_{i\in [0, n)}$ which satisfy
$
({\bf A}[0]/{\bf A}^\prime[0], \ldots, {{\bf A}[n-1]}/{{\bf A}^\prime[n-1]}) \neq {\bf 1}_{\mathbb{G}_1} \wedge {1}_{\mathbb{G}_{\rm T}}=\prod_{i=0}^{n-1}e({\bf A}[i]/{\bf A}^\prime[i], h^{\beta^{2i}})
$ with non-negligible probability, breaking $n$-ASDBP assumption.

\section{Correctness and Security Analysis of FlexProofs}\label{app:vc}
\begin{theorem}
	FlexProofs satisfies the correctness property (Definition  \ref{df:VCCor}).
\end{theorem}
\begin{proof}
	We first prove that the interactive version of FlexProofs satisfies the correctness property.
	Specifically, we show that for any security parameter $\lambda$, any integer $N>0$, any vector ${\bf m}\in \mathbb{F}_p^N$, and any index $i\in [0, N)$, it holds that:
	$$
	\Pr
	\begin{bmatrix}
	{\sf pp} \leftarrow {\sf VC.Setup}(1^\lambda, 1^{N}),
	(C, {\sf aux})\leftarrow {\sf VC.Commit}({\bf m}),\\
	\langle {\cal P}_{\it OpenAll}({\sf pp}, {\sf aux}, i, {\bf m}), ~{\cal V}_i({\sf pp}, C, i, {\bf m}[i])\rangle=1
	\end{bmatrix}=1.
	$$
	By Eq.~\eqref{eq:genCi} and \eqref{eq:genC}, $C$ is a commitment to a vector ${\bf C}(=(C_0, \ldots, C_{\mu-1}))$ under the FC scheme. 
	Let $k = \lfloor i / (\mu b) \rfloor$ and $K=[kb, (k+1)b)$.
	According to Eq.~\eqref{eq:genFCBatchProof}, the proof $\pi_{{\bf C}_k}$ is a batch proof for $\{C_a=\langle {\bf C}, {\bf u}_a\rangle\}_{a\in K}$ (i.e., $C_a$ is the $a$-th element of ${\bf C}$). 
	By Definition \ref{df:FCCor}, the correctness property of the FC scheme ensures that ${\sf FC.BVerify}({C}, \{{\bf u}_a\}_{a\in K}, \{C_a\}_{a\in K}, \pi_{{\bf C}_{k}})=1.$
	
	As per Eq. \eqref{eq:ProverRand}, the prover ${\cal P}_{OpenAll}$ random its values.
	For convenience, let $[i]_\mu=i \bmod \mu$.
	As per Eq. \eqref{eq:VerifierRand}, the verifier ${\cal V}_i$ computes $D_{\lfloor i/\mu \rfloor}$ and ${y}_{\lfloor i/\mu \rfloor, [i]_\mu}$. 
	As described in Step 2.2, the verifier also receives 
	$\{D_w, y_{w, [i]_\mu}\}_{w \in {\sf sib}(\lfloor i/\mu \rfloor)}$, where ${\sf sib}(\lfloor i/\mu \rfloor)$ is the set of sibling nodes along the path from the root to the $\lfloor i/\mu \rfloor$-th leaf in the folding tree. Then, the verifier computes:
	$$
	D^* = \sum_{w \in {\sf sib}(\lfloor i/\mu \rfloor)} D_w + D_{\lfloor i/\mu \rfloor}, \quad
	y^*_{[i]_\mu} = \sum_{w \in {\sf sib}(\lfloor i/\mu \rfloor)} y_{w, [i]_\mu} + {y}_{\lfloor i/\mu \rfloor, [i]_\mu}.
	$$	
	By the construction of the folding scheme, we have:
	$D^* = \prod_{j=0}^{\mu-1} C_j^{r_j}$ and $g^* = \sum_{j=0}^{\mu-1} r_j f_j,
	$
	so that $D^*$ is a commitment to the folded polynomial $g^*$, and $y^*_{[i]_\mu} = g^*({\sf Bin}([i]_\mu))$.
	As per Step 2.3, the proof $\pi^*_{[i]_\mu}$ proves $y^*_{[i]_\mu} = g^*({\sf Bin}([i]_\mu))$.
	Then, by Definition \ref{df:PCCom}, the correctness property of the underlying PC scheme implies that ${\sf PC.Verify}(D^*, {\sf Bin}({[i]_\mu}), y_{[i]_\mu}^*, \pi^*_{[i]_\mu})=1.$
	
	Thus, the verifier accepts the proof with probability 1, confirming the correctness of the interactive protocol.
	And under the random oracle model, the correctness of the non-interactive protocol resulting from the Fiat-Shamir transformation \cite{ZXH+22} follows directly from the established correctness of the underlying interactive scheme.
\end{proof}
\begin{theorem}
	FlexProofs satisfies the position binding property (Definition \ref{df:VCSd}) in the ROM.
\end{theorem}
\begin{proof}
	We first prove that the interactive version of FlexProofs satisfies the position binding property.
	Concretely, we show that for any security parameter $\lambda$, any integer $N>0$, and any PPT adversary $\mathcal{A}_{OpenAll}$,
	$$
	\Pr
	\begin{bmatrix}
	{\sf pp} \leftarrow {\sf VC.Setup}(1^\lambda, 1^{N}):
	(\langle {\cal A}_{\it OpenAll}, {\cal V}_i\rangle({\sf pp}, C, i, m_i) =1) \wedge\\
	(\langle {\cal A}_{\it OpenAll}, {\cal V}_i\rangle({\sf pp}, C, i, m_i^\prime) =1) \wedge
	(m_i\neq m_i^\prime)
	\end{bmatrix}\leq {\sf negl}(\lambda).
	$$
	Suppose there is a PPT adversary $\mathcal{A}_{OpenAll}$ that breaks the position binding property of the interactive version with a non-negligible probability $\epsilon$. 
	To distinguish the messages, denote all messages and proofs during the second acceptance (for $m_i'$) with primes.
	
	Let $k = \lfloor i / (\mu b) \rfloor$ and $K=[kb, (k+1)b)$.
	Then either
	\begin{align}
		&\exists a\in K:C_a\neq C_a^\prime,\label{eq:case1} \text{ or }\\
		&\forall a\in K:C_a=C_a^\prime.\label{eq:case2}
	\end{align}
	For the remaining part of the proof, we discuss two cases. 
	In the first case, $\mathcal{A}_{OpenAll}$ breaks the position binding property of the interactive version with two accepting proof transcripts that satisfies \eqref{eq:case1}.  
	In the second case, $\mathcal{A}$ breaks the position binding property of the interactive version with two accepting proof transcripts that satisfies \eqref{eq:case2}. 
	For $\ell\in\{1,2\}$, let $\epsilon_\ell$ be   the probability that case $\ell$ occurs. Then 
	$\epsilon=\epsilon_1+\epsilon_2$.	
	We will show that $\epsilon$ is negligible by showing that both $\epsilon_1$ and $\epsilon_2$ are negligible.
	
	\vspace{2mm}
	\noindent
	{\bf $\epsilon_1$ is negligible.}
	If $\epsilon_1$ is non-negligible, we 
	construct an adversary $\mathcal{B}_{FC}$ that breaks the function binding of the FC scheme with probability at least $\epsilon_1$ and thus give a contradiction. 
	Given the ${\sf pp}_{FC}$ in Eq. \eqref{eq:vcpp}, 
	$\mathcal{B}_{FC}$ generates ${\sf pp}_{PC}$ as per Eq. \eqref{eq:vcpp}, and invokes $\mathcal{A}_{OpenAll}$ with ${\sf pp}=\{{\sf pp}_{FC}, {\sf pp}_{PC}\} $.
	Upon receiving $\mathcal{A}_{OpenAll}$’s two accepting proof transcripts that satisfies \eqref{eq:case1}, $\mathcal{B}_{FC}$ outputs
	$$
	(C, \{{\bf u}_a, C_a, C_a^\prime\}_{a\in K}, \pi_{{\bf C}_{k}}, \pi_{{\bf C}_{k}}^\prime).
	$$ 
	Note that these values satisfies the following properties:
	\begin{enumerate}
		\item[(1)]
		${\sf FC.BVerify}({C}, \{{\bf u}_a\}_{a\in K}, \{C_a\}_{a\in K}, \pi_{{\bf C}_{k}})=1$; (by the transcript for $m_i$)
		\item[(2)]
		${\sf FC.BVerify}({C}, \{{\bf u}_a\}_{a\in K}, \{C_a^\prime\}_{a\in K}, \pi_{{\bf C}_{k}}^\prime)=1$; (by the transcript for $m_i^\prime$)
		\item[(3)]
		$\exists a\in K:C_a\neq C_a^\prime$. (by Eq. \eqref{eq:case1})
	\end{enumerate}
	Therefore, $\mathcal{B}_{FC}$ breaks the function binding property of the FC scheme with probability $\geq \epsilon_1$, which gives a contradiction. 
	
	\vspace{2mm}
	\noindent
	{\bf $\epsilon_2$ is negligible.}
	The soundness of the folding scheme for ${\sf PC.HyperEval}$ is guaranteed by \cite{GWC19,BDF+21}.
	If $\epsilon_2$ is non-negligible, we 
	construct an adversary $\mathcal{B}_{FH}$ that breaks the soundness of the folding scheme with probability at least $\epsilon_2$ and thus give a contradiction. 
	
	For convenience, let $j={\lfloor i/\mu \rfloor}$ and $[i]_\mu=i \bmod \mu$.
	First, recall that the accepting proof transcript contains the following: (1) a random value $r_j$, (2) $\{D_w, y_{w, [i]_\mu}\}_{w\in {\sf sib}(j)}$, where ${\sf sib}(j)$ denotes the set of sibling nodes along the path from the root to the $j$-th leaf node in the folding tree and (3) an opening proof $\pi^*_{[i]_\mu}$ of the final folded polynomial.
	
	Given the ${\sf pp}_{PC}$ in Eq. \eqref{eq:vcpp}, 
	$\mathcal{B}_{FH}$ generates ${\sf pp}_{FC}$ as per Eq. \eqref{eq:vcpp}, and invokes $\mathcal{A}_{OpenAll}$ with ${\sf pp}=\{{\sf pp}_{FC}, {\sf pp}_{PC}\}$.
	Upon receiving $\mathcal{A}_{OpenAll}$’s two accepting proof transcripts that satisfies \eqref{eq:case1}, $\mathcal{B}_{FH}$
	computes $$D_{c}=\sum_{w\in {\sf sib}(j)} D_w, ~y_{[i]_\mu}=\sum_{w\in {\sf sib}(j)} y_{w, [i]_\mu}, ~D_{c}^\prime=\sum_{w\in {\sf sib}(j)} D_w^\prime, ~y_{[i]_\mu}^\prime=\sum_{w\in {\sf sib}(j)} y_{w, [i]_\mu}^\prime.$$
	Then it outputs two folding transcripts
	$$
	\rho_1=(D_c, C_j, [i]_\mu, y_{[i]_\mu}, m_i, r_j, \pi^*_{[i]_\mu}), 
	\rho_2=(D_c^\prime, C_j^\prime, [i]_\mu, y_{[i]_\mu}^\prime, m_i^\prime, r_j^\prime, \pi^{*\prime}_{[i]_\mu}).
	$$
	Note that these values satisfies the following properties:
	\begin{enumerate}
		\item[(1)]
		$\rho_1$ passes the verification of the folding scheme; (by the transcript for $m_i$)
		\item[(2)]
		$\rho_2$ passes the verification of the folding scheme; (by the transcript for $m_i^\prime$)
		\item[(3)]
		$C_j=C_j^\prime$; (by Eq. \eqref{eq:case2}) 
		\item[(4)]
		$m_i\neq m_i^\prime$. 
	\end{enumerate}
	Therefore, $\mathcal{B}_{FH}$ breaks the soundness of the folding scheme for ${\sf PC.HyperEval}$	with probability $\geq \epsilon_2$, which gives a contradiction. 	
	
	We have shown that the interactive version of FlexProofs satisfies soundness, captured by the {\em position binding} property. 
	Next, we show that the interactive version of FlexProofs satisfies a stronger notion of soundness called {\em state restoration soundness} \cite{BCS16}.
	To prove state restoration soundness, we consider adversaries that are allowed to reset or rewind the verifier, potentially reusing the verifier's randomness or internal state. We note that the above soundness proof is robust to such behaviors. In particular, the reductions used in both cases (function binding and folding soundness) remain valid even if the verifier's challenge randomness is reused across different transcripts. Hence, the soundness proof already implies that the interactive protocol satisfies state restoration soundness.
	
	The work \cite{ZXH+22} formally proves that 
	as long as an interactive proof protocol is secure against
	state restoration attacks, then the non-interactive	protocol by applying their Fiat-Shamir transformation is sound in the random oracle model.
	Therefore, we conclude that our non-interactive scheme satisfies soundness (i.e., position binding) in the random oracle model.
	\qed
	%Finally, because the Fiat-Shamir transformation we use is identical to \cite{ZXH+22}, we can apply their techniques to prove the soundness of the non-interactive variant of our VC scheme.
\end{proof}
\end{document}